\pdfoutput=1
\RequirePackage{ifpdf}
\ifpdf 
\documentclass[pdftex]{sigma}
\else
\documentclass{sigma}
\fi

\usepackage{tikz}
\usepackage{tikzscale}
\usetikzlibrary{matrix,shapes,decorations.pathreplacing,fit,backgrounds}
\usepackage{mathtools}
\usepackage{pgfplots}
\pgfplotsset{width=5.5in, compat=1.14}
\usepackage{pgf,tikz,pgfplots}
\pgfplotsset{compat=1.15}
\newcommand*\circled[1]{\tikz[baseline=(char.base)]{
 \node[shape=circle,draw,inner sep=2pt] (char) {#1};}}

\newcommand{\Rmnum}[1]{\expandafter\@slowromancap\romannumeral #1@}

\def\Xint#1{\mathchoice
{\XXint\displaystyle\textstyle{#1}}%
{\XXint\textstyle\scriptstyle{#1}}%
{\XXint\scriptstyle\scriptscriptstyle{#1}}%
{\XXint\scriptscriptstyle\scriptscriptstyle{#1}}%
\!\int}
\def\XXint#1#2#3{{\setbox0=\hbox{$#1{#2#3}{\int}$}
\vcenter{\hbox{$#2#3$}}\kern-.5\wd0}}
\def\dashint{\Xint-}

\numberwithin{equation}{section}

\newtheorem{Theorem}{Theorem}[section]
\newtheorem{Corollary}[Theorem]{Corollary}
\newtheorem{Lemma}[Theorem]{Lemma}
\newtheorem{Proposition}[Theorem]{Proposition}
 { \theoremstyle{definition}
\newtheorem{Example}[Theorem]{Example}
\newtheorem{Remark}[Theorem]{Remark} }

\begin{document}

\newcommand{\arXivNumber}{2104.06471}

\renewcommand{\PaperNumber}{008}

\FirstPageHeading

\ShortArticleName{Simplified Forms of the Transition Probabilities of the Two-Species ASEP}

\ArticleName{Simplified Forms of the Transition Probabilities\\ of the Two-Species ASEP with Some Initial Orders\\ of Particles}

\Author{Eunghyun LEE and Temirlan RAIMBEKOV}
\AuthorNameForHeading{E.~Lee and T.~Raimbekov}
\Address{Department of Mathematics, Nazarbayev University, Nur-sultan, Kazakhstan}
\Email{\href{mailto:eunghyun.lee@nu.edu.kz}{eunghyun.lee@nu.edu.kz}, \href{mailto:temirlan.raimbekov@alumni.nu.edu.kz}{temirlan.raimbekov@alumni.nu.edu.kz}}
\URLaddress{\url{https://sites.google.com/a/nu.edu.kz/eunghyun-lee-s-homepage/}}

\ArticleDates{Received April 15, 2021, in final form January 24, 2022; Published online January 29, 2022}

\Abstract{It has been known that the transition probability of the single species ASEP with $N$ particles is expressed as a sum of $N!$ $N$-fold contour integrals which are related to permutations in the symmetric group $S_N$. On other hand, the transition probabilities of the multi-species ASEP, in general, may be expressed as a sum of much more terms than $N!$. In this paper, we show that if the initial order of species is given by $2\cdots 21$, $12\cdots 2$, $1\cdots 12$ or $21\cdots 1$, then the transition probabilities can be expressed as a sum of at most $N!$ contour integrals, and provide their formulas explicitly.}


\Classification{82C22; 60J27}

\section{Introduction}
\subsection{Definition of the model and previous results}
In the multi-species asymmetric simple exclusion process (ASEP) on $\mathbb{Z}$, particles belong to a species labelled by one of positive integers. Each particle chooses a direction to move one step to the right or to the left after waiting time exponentially distributed with rate 1. The probability to choose the right direction is $p$ and the probability to choose the left direction is $q=1-p$. If the target site is empty, the particle moves to the site, but if the site is already occupied, the following rule is applied: if a particle belonging to species $l$ tries moving to the target site occupied by a particle belonging to species $l'\geq l$, then the move is prohibited, but if $l'<l$, the particle belonging to $l$ can move to the target site by interchanging sites with the particle belonging to $l'$. If we assume that there are $N$ particles, a state is denoted by a~pair $(X, \pi)$ where $X = (x_1,\dots, x_N) \in \mathbb{Z}^N$ with $x_1<\cdots <x_N$ for the positions of particles and $\pi= \pi(1)\pi(2)\cdots\pi(N)$ is a permutation of a multi-set $\mathcal{M} = [i_1,\dots,i_N]$ with elements taken from $\{1,\dots, N\}$ to represent the species of particles. Here, $\pi(i)$ represents the $i^{\rm th}$ leftmost particle's species. Since particles can interchange their positions, the order of species may change over time. The transition probability from the initial state $(Y,\nu)$ to state $(X,\pi)$ at time $t$ is denoted by $P_{(Y,\nu)}(X,\pi;t).$ Earlier works on the transition probabilities and some distributions in the multi-species ASEP are found in \cite{Chatterjee-Schutz-2010,Kuan-2018,Kuan-2020,Lee-2017,Lee-2018,Lee-2020,Tracy-Widom-2009,Tracy-Widom-2013}. Also, the multi-species ASEP can be considered as a special case of the coloured stochastic vertex model (see \cite{Borodin-Bufetov,Borodin-Wheeler}). According to~\cite{Lee-2020}, the transition probability of the multi-species ASEP with $N$ particles from the initial state $(Y,\nu)$ to the state $(X,\pi)$ is a matrix element of an $N^N \times N^N$ matrix $ \mathbf{P}_Y(X;t)$ whose columns $\nu$ and rows $\pi$ are labelled by $11\cdots 1, \dots, NN\cdots N$. In \cite{Lee-2020}, the formula of $ \mathbf{P}_Y(X;t)$ is written
\begin{gather}\label{1223-am-51911}
 \mathbf{P}_Y(X;t) = \dashint_c\cdots \dashint_c\sum_{\sigma\in {S}_N}\mathbf{A}_{\sigma}\prod_{i=1}^N\big(\xi_{\sigma(i)}^{x_i-y_{\sigma(i)}-1}{\rm e}^{\varepsilon(\xi_i) t}\big)\, {\rm d}\xi_1\cdots {\rm d}\xi_N,
\end{gather}
where $\mathbf{A}_{\sigma}$ is a $N^N \times N^N$ matrix. Note that the form of (\ref{1223-am-51911}) resembles~(2.3) in~\cite{Tracy-Widom-2008}. Here, $\dashint_c$ implies $(1/2\pi {\rm i})\int_c$ where the contour~$c$ is a positively oriented circle centered at the origin with sufficiently small radius so that no poles except the one at the origin are included in~$c$, and
\begin{gather*}
\varepsilon(\xi_i) = \frac{p}{\xi_i} +q\xi_i -1
\end{gather*}
and the sum in (\ref{1223-am-51911}) is over all permutations $\sigma$ in the symmetric group $S_N$.
The integral on the right-hand side of (\ref{1223-am-51911}) implies the matrix element-wise integral, so the $(\pi,\nu)^{\rm th}$ term of $\mathbf{P}_Y(X;t)$ is
\begin{align}
[\mathbf{P}_Y(X;t)]_{\pi,\nu}& = P_{(Y,\nu)}(X,\pi;t)\nonumber\\
 & =\dashint_c\cdots \dashint_c\sum_{\sigma\in {S}_N} [\mathbf{A}_{\sigma} ]_{\pi,\nu}\prod_{i=1}^N\big(\xi_{\sigma(i)}^{x_i-y_{\sigma(i)}-1}{\rm e}^{\varepsilon(\xi_i) t}\big)\, {\rm d}\xi_1\cdots {\rm d}\xi_N.\label{1050pm102}
 \end{align}
The procedure introduced in \cite{Lee-2020} to find the matrix $\mathbf{A}_{\sigma}$ is as follows. Let $T_i$ be the simple transposition which interchanges the number at the $i^{\rm th}$ slot and the number at the $(i+1)^{\rm st}$ slot. It is well known that simple transpositions $T_1,\dots, T_{N-1}$ generate the symmetric group $S_N$. Hence, any permutation $\sigma \in S_N$ can be written as a product of simple transpositions, that is,
\begin{gather}\label{853pm35}
\sigma = T_{i_j}\cdots T_{i_1}
\end{gather}
for some $i_1,\dots, i_j \in \{1,\dots, N-1\}$. (Here, the expression (\ref{853pm35}) is not unique.) If $T_i$ is acted on a permutation with $\alpha$ at the $i^{\rm th}$ slot and $\beta$ at the $(i+1)^{\rm st}$ slot, then
\begin{gather*}
T_i(\cdots \alpha\beta \cdots) = (\cdots \beta\alpha \cdots),
\end{gather*}
and we denote this $T_i$ by $T_i(\beta,\alpha)$ when we need to show explicitly which numbers are interchanged. For example, for $N=3$,
\begin{gather*}
T_1T_2T_1 = T_1T_2T_1(123) = T_1T_2(213)=T_1(231) = 321,
\end{gather*}
and we write $321 = T_1(3,2)T_2(3,1)T_1(2,1)$. Hence, we will write (\ref{853pm35}) as
\begin{gather}\label{123am107}
\sigma = T_{i_j}(\beta_j,\alpha_j)\cdots T_{i_1}(\beta_1,\alpha_1)
\end{gather}
 when necessary.
 Corresponding to $T_{i}(\beta,\alpha)$, we define $N^N \times N^N$ matrix $\mathbf{T}_{i}(\beta,\alpha)$ by
\begin{gather*}
\mathbf{T}_{i}(\beta,\alpha) =\underbrace{\mathbf{I}_N \otimes \cdots \otimes \mathbf{I}_N}_{(i-1)~\text{times}} \otimes \mathbf{R}_{\beta\alpha} \otimes \underbrace{\mathbf{I}_N \otimes \cdots \otimes \mathbf{I}_N}_{(N-i-1)~\text{times}},
\end{gather*}
where $\mathbf{I}_N$ is the $N \times N$ identity matrix and $\mathbf{R}_{\beta\alpha}$ is an $N^2 \times N^2$ matrix whose columns and rows are labelled by $11, 12, \dots, 1N, 21,\dots, NN$. The matrix elements of $\mathbf{R}_{\beta\alpha}$ are given by
\begin{gather}\label{625pm72443}
[\mathbf{R}_{\beta\alpha}]_{ij,kl} = \begin{cases}
S_{\beta\alpha}& \text{if}~ij=kl~\text{with}~i=j,\\
P_{\beta\alpha}& \text{if}~ij=kl~\text{with}~i<j,\\
Q_{\beta\alpha}& \text{if}~ij=kl~\text{with}~i>j,\\
pT_{\beta\alpha}&\text{if}~ij=lk~\text{with}~i<j,\\
qT_{\beta\alpha}&\text{if}~ij=lk~\text{with}~i>j,\\
0 &\text{for all other cases},
\end{cases}
\end{gather}
where
\begin{alignat*}{3}
& S_{\beta\alpha} =-\frac{p+q\xi_{\alpha}\xi_{\beta} - \xi_{\beta}}{p+q\xi_{\alpha}\xi_{\beta} - \xi_{\alpha}},\qquad && P_{\beta\alpha} = \frac{(p-q\xi_{\alpha})(\xi_{\beta}-1)}{p+q\xi_{\alpha}\xi_{\beta} - \xi_{\alpha}},& \nonumber\\
 & T_{\beta\alpha} = \frac{\xi_{\beta}-\xi_{\alpha}}{p+q\xi_{\alpha}\xi_{\beta} - \xi_{\alpha}},\qquad && Q_{\beta\alpha} =\frac{(p-q\xi_{\beta})(\xi_{\alpha}-1)}{p+q\xi_{\alpha}\xi_{\beta} - \xi_{\alpha}}.&
\end{alignat*}
With this setting, it was obtained that $\mathbf{A}_{\sigma}$ is given by
\begin{gather}\label{932pm35}
\mathbf{A}_{\sigma} = \mathbf{T}_{i_j}(\beta_j,\alpha_j) \cdots \mathbf{T}_{i_1}(\beta_1,\alpha_1),
\end{gather}
when $\sigma$ is written as in (\ref{123am107}). Although the expression (\ref{123am107}) is not unique, (\ref{932pm35}) is well-defined in the sense that (\ref{932pm35}) represents the same matrix for any expression (\ref{123am107}) (see \cite[Remark~2.2]{Lee-2020}).

\subsection{Motivation and main results}
\subsubsection{Motivation}\label{656pm127}
It can be shown that $[\mathbf{A}_{\sigma}]_{\pi,\nu} = 0$ in (\ref{1050pm102}) unless both $\pi$ and $\nu$ are from the same multi-set (see Section~\ref{454pm321}), and in this case, it is obvious that $P_{(Y,\nu)}(X,\pi;t) = 0$, which physically also makes sense. However, in general, the matrix element of $\mathbf{A}_{\sigma}$ is written as a~\textit{sum} because $\mathbf{A}_{\sigma}$ is a~product of matrices as seen in~(\ref{932pm35}). If we want to know the explicit formula of $P_{(Y,\nu)}(X,\pi;t)$, we should know the corresponding matrix elements, $[\mathbf{A}_{\sigma}]_{\pi,\nu}$, explicitly, for all $\sigma$. It is interesting that some elements of~$\mathbf{A}_{\sigma}$ that are expressed as \textit{sums} can be further factorized. For example, for $\sigma = 321 = T_2(2,1)T_1(3,1)T_2(3,2)$,
\begin{gather*}
 [\mathbf{A}_{321}]_{121,211} = [\mathbf{T}_2(2,1)\mathbf{T}_1(3,1)\mathbf{T}_2(3,2) ]_{121,221} = pP_{32}Q_{31}T_{21}+pQ_{21}S_{31}T_{32}
\end{gather*}
but we can further observe that
\begin{gather*}
pP_{32}Q_{31}T_{21}+pQ_{21}S_{31}T_{32}= Q_{21}pT_{31}S_{32}.
\end{gather*}
Actually, if we use a different expression $\sigma = 321 = T_1(3,2)T_2(3,1)T_1(2,1)$, then we directly obtain
\begin{gather*}
 [\mathbf{A}_{321}]_{121,211} = [\mathbf{T}_1(3,2)\mathbf{T}_2(3,1)\mathbf{T}_1(2,1) ]_{121,221} = Q_{21}pT_{31}S_{32}
\end{gather*}
through matrix multiplication. However, unfortunately, not all $ [\mathbf{A}_{\sigma}]_{\pi,\nu}$ can be simplified to a~factorized form. In this paper, we are interested in the factorized forms of $[\mathbf{A}_{\sigma}]_{\pi,\nu} $ because single-species models for which the Bethe ansatz can be applicable have factorized forms in the transition probabilities and some interesting probability distributions were obtained from these transition probabilities \cite{Korhonen-Lee-2014,Lee-2010,Lee-2012,Lee-2020,Lee-Wang-2019,Nagao, Rakos-Schutz-2006,Tracy-Widom-2008,Tracy-Widom-2018}. In this paper, we show that if the initial permutation $\nu$ of species is one of $\nu = 2\cdots 21$, $1\cdots 12$, $21\cdots 1$ and $12\cdots 2$, and $\sigma$'s expression by simple transpositions is obtained in a special way, then $[\mathbf{A}_{\sigma}]_{\pi,\nu}$ is zero or written as a factorized form
\begin{gather}\label{627am411}
[\mathbf{A}_{\sigma}]_{\pi,\nu} =\prod_{(\beta,\alpha)}R_{\beta\alpha},
\end{gather}
where $R_{\beta\alpha}$ is one of $S_{\beta\alpha}$, $Q_{\beta\alpha}$, $P_{\beta\alpha}$, $pT_{\beta\alpha}$ or $qT_{\beta\alpha}$, which is similar to
\begin{gather*}
[\mathbf{A}_{\sigma}]_{1\cdots1,1\cdots 1} = \prod_{(\beta,\alpha)}S_{\beta\alpha}
\end{gather*}
in the (single species) ASEP \cite{Tracy-Widom-2008}. The notation $\prod_{(\beta,\alpha)}$ implies that the product is taken over all inversions $(\beta,\alpha)$ of $\sigma$. An inversion of a permutation $\sigma=\sigma(1)\cdots\sigma(N)$ is a pair of elements $(\sigma(i), \sigma(j))$ with $i<j$ and $\sigma(i)>\sigma(j)$.
It will remain for future works to see if some interesting probability distributions can be obtained in \textit{neat} forms from the transition probabilities with the formulas of $[\mathbf{A}_{\sigma}]_{\pi,\nu}$ given in this paper.

\subsubsection{Main results}\label{1126pm1010}
Each permutation in $S_N$ may be written as a product of simple transpositions in many ways via
\begin{gather}\label{952pm411}
\begin{aligned}
& T_iT_j = T_jT_i&~~~\text{if $|i-j| \geq 2$},\\
& T_iT_jT_i = T_jT_iT_j&~~~\text{if $|i-j| = 1$},\\
& T_i^2 =1.
\end{aligned}
\end{gather}
One of the findings in this paper is that if $\sigma$ is expressed as in Theorem \ref{700pm410}, which is a known fact, then $[\mathbf{A}_{\sigma}]_{\pi, 2\cdots 21}$ is zero or written as (\ref{627am411}). If $\sigma$ is expressed in a different way from Theorem~\ref{700pm410}, then the form of $[\mathbf{A}_{\sigma}]_{\pi, 2\cdots 21}$ may not be in a factorized form.
\begin{Theorem}[\cite{Kassel}]\label{700pm410}
Consider the following subsets of the symmetric group $S_N$.
\begin{gather*}
\begin{aligned}
&\Sigma_1 = \{1,T_1\}, \\
&\Sigma_2 = \{1,T_2,T_2T_1\}, \\
&\Sigma_3 = \{1,T_3,T_3T_2,T_3T_2T_1\}, \\
&\hspace{0.3cm} \vdots \\
&\Sigma_{N-1} = \{1,T_{N-1},T_{N-1}T_{N-2},\dots, T_{N-1}\cdots T_2T_1\}.
\end{aligned}
\end{gather*}
For any permutation $\sigma \in S_N$, there is a unique element
\begin{gather*}
(w_1,\dots, w_{N-1}) \in \Sigma_1 \times \Sigma_2 \times \cdots \times \Sigma_{N-1}
\end{gather*}
such that $\sigma = w_1w_2\cdots w_{N-1}$.
\end{Theorem}
\begin{Theorem}\label{420am36}
Let $\sigma =w_1\cdots w_{N-1}$ be an expression as in Theorem~{\rm \ref{700pm410}} and let
\begin{gather*}
\mathbf{A}_{\sigma} = \mathbf{T}_{i_j}(\beta_j,\alpha_j)\cdots\mathbf{T}_{i_2}(\beta_2,\alpha_2) \mathbf{T}_{i_1}(\beta_1,\alpha_1)
\end{gather*}
be the matrix corresponding to $\sigma$.
Then, for all $N \geq 2$,
\begin{itemize}\itemsep=0pt
 \item [$(a)$] $ \big[{\mathbf{A}}_{12\cdots N}\big]_{2\cdots 21,2\cdots 21} = 1 $ and $ \big[{\mathbf{A}}_{12\cdots N}\big]_{\pi,2\cdots 21} = 0 $ if $\pi \neq 2\cdots 21$.
 \item [$(b)$] If $\sigma \neq 12\cdots N$, then $\big[{\mathbf{A}}_{\sigma}\big]_{\pi,2\cdots 21}$ is zero or written as
 \begin{gather}\label{428am36}
 \big[{\mathbf{A}}_{\sigma}\big]_{\pi,2\cdots 21} = \prod_{(\beta,\alpha)}R_{\beta\alpha}
\end{gather}
where $R_{\beta\alpha}$ is one of $S_{\beta\alpha}$, $P_{\beta\alpha}$, $Q_{\beta\alpha}$, $pT_{\beta\alpha}$, $qT_{\beta\alpha}$. The product in~\eqref{428am36} is taken over all inversions $(\beta,\alpha)$ in $\sigma$.
\end{itemize}
\end{Theorem}
\begin{Remark} Theorem \ref{420am36}(b) asserts the existence of the factorized form but does not provide information on how to choose the form of~$R_{\beta\alpha}$. The choice for $R_{\beta\alpha}$ in~(\ref{428am36}) depends on~$\sigma$ and~$\pi$. The explicit form of~(\ref{428am36}) will be provided in Theorems~\ref{1013pm104} and~\ref{312am109}.
\end{Remark}
 The method to express $\sigma$ in Theorem~\ref{700pm410} does not work for other initial orders $\nu= 12\cdots2,$ $1\cdots 12,$ $21\cdots 1$. In other words, $[\mathbf{A}_{\sigma}]_{\pi,\nu}$ may not be in a factorized form if $\nu\neq 2\cdots 21$. For $\nu= 12\cdots2,$ $1\cdots 12,$ $21\cdots 1$, we need different methods to express $\sigma$ by simple transpositions for factorized forms of $[{\mathbf{A}}_{\sigma}]_{\pi,\nu}$. These methods will be provided in Sections~\ref{1034pm106} and~\ref{1035pm106}.

Now, we state the formulas for $[{\mathbf{A}}_{\sigma}]_{\pi,\nu}$. First, we find the formula of the diagonal terms. That is, this is the case that initial order of species and the order at time $t$ are the same.
\begin{Theorem}\label{1013pm104}
Let $\nu=2\cdots 21$ and
\begin{gather*}
\sigma = w_1\cdots w_{N-1} = {T}_{i_j}(\beta_j,\alpha_j)\cdots {T}_{i_1}(\beta_1,\alpha_1)
\end{gather*}
be given as in Theorem~{\rm \ref{700pm410}}. Let
\begin{gather*}
\mathbf{A}_{\sigma}=\mathbf{T}_{i_j}(\beta_j,\alpha_j)\cdots \mathbf{T}_{i_1}(\beta_1,\alpha_1)
\end{gather*}
 be the matrix corresponding to $\sigma$. Then,
 \begin{align}
[\mathbf{A}_{\sigma}]_{\nu, \nu}& = [\mathbf{T}_{i_j} ]_{\nu, \nu} \cdots [\mathbf{T}_{i_m} ]_{\nu, \nu} \cdots [\mathbf{T}_{i_1} ]_{\nu, \nu} \nonumber\\
& =\begin{cases}
S_{\beta_{j}\alpha_{j}}\cdots S_{\beta_{m+1}\alpha_{m+1}} Q_{\beta_m\alpha_m}S_{\beta_{m-1}\alpha_{m-1}}\cdots S_{\beta_1\alpha_1}& \text{if~$i_m = N-1$},\\
S_{\beta_{j}\alpha_{j}}\cdots S_{\beta_1\alpha_1}& \text{if~$i_i,\dots, i_j \neq N-1$}.
 \end{cases}\label{6166pm104}
\end{align}
\end{Theorem}
Recall that $T_{N-1}$ appears at most once in $\sigma = w_1\cdots w_{N-1}$ in Theorem \ref{700pm410}. In~(\ref{6166pm104}), the appearance of $Q_{\beta\alpha}$ depends on the existence of $T_{N-1}$ in $\sigma$.
\begin{Remark}
We will provide the formula for $[\mathbf{A}_{\sigma}]_{\nu, \nu}$ with $\nu=12\cdots2,$ $1\cdots 12,$ $21\cdots 1$ in Section~\ref{300pm106}. Their formulas are similar to~(\ref{6166pm104}). The formulas of $[\mathbf{A}_{\sigma}]_{\nu, \nu}$ for some $\nu$ in the two-species TASEP are found in \cite[Lemma~2.2]{Lee-2017} and \cite[Lemma~3.3]{Lee-2018}. Moreover, it was shown that the transition probabilities of the multi-species TASEP are expressed as a determinant when the initial order and the order at time $t$ are the same (see \cite[Theorem~3.1]{Lee-2020}). In case of the totally asymmetric model, the matrices $\mathbf{T}$ in (\ref{932pm35}) are upper-triangular, so $[\mathbf{A}_{\sigma}]_{\nu,\nu}$ is expressed as a product of some factors over all inversions in~$\sigma$ for any~$\nu$.
\end{Remark}
If $\pi \neq \nu$, it is possible that $ [\mathbf{A}_{\sigma}]_{\pi, \nu} = 0$. Let $\pi^{(i)}$ be the permutation on the multi-set $[1,2,\dots,2]$ with $\pi(i) = 1$, for example, $\pi^{(1)}= 12\cdots 2$ and $\pi^{(N)}= 2\cdots 21=\nu.$

\begin{Theorem}\label{302pm106}
Let $\sigma= w_1\cdots w_{N-1}$ be expressed as in Theorem~{\rm \ref{700pm410}}. If $l$ is the largest integer such that $w_l = 1$, then $ [\mathbf{A}_{\sigma}]_{\pi^{(i)}, \nu} = 0$ for all $i=1,\dots,l$ and $ [\mathbf{A}_{\sigma}]_{\pi^{(i)}, \nu} \neq 0$ for all $i=l+1,\dots, N$. Moreover, if $w_i\neq 1$ for all $i$ in $\sigma= w_1\cdots w_{N-1}$, then $ [\mathbf{A}_{\sigma}]_{\pi^{(i)}, \nu} \neq 0$ for all $i$.
\end{Theorem}

According to Theorem \ref{302pm106}, if $w_{N-1} = 1$, then $[\mathbf{A}_{\sigma}]_{\pi, 2\cdots 21} = 0$ for all $\pi \neq 2\cdots 21$. Now, we give the most general formula of $ [\mathbf{A}_{\sigma}]_{\pi^{(i)}, \nu}$ with $\nu=2\cdots 21$. First, we introduce some notations. In the expression $\sigma= w_1\cdots w_{N-1}$ given by Theorem \ref{700pm410}, each $w_i \in \Sigma_i$ is 1 or a product of simple transpositions, that is,
\begin{gather*}
w_i = T_i(\beta_i,\alpha_i)T_{i-1}(\beta_{i-1},\alpha_{i-1})\cdots T_l(\beta_l,\alpha_l)
\end{gather*}
for some $l$. We write
\begin{gather*}
\mathbf{w}_i = \mathbf{T}_i(\beta_i,\alpha_i)\mathbf{T}_{i-1}(\beta_{i-1},\alpha_{i-1})\cdots \mathbf{T}_l(\beta_l,\alpha_l)
\end{gather*}
for the matrix corresponding to $w_i$. We define
\begin{gather*}
\prod_{i = l}^m A_i =1 \quad \text{if~$l>m$}.
\end{gather*}
\begin{Theorem}\label{312am109}
Let $\sigma= w_1\cdots w_{N-1}$ be expressed as in Theorem~{\rm \ref{700pm410}}. Then,
\begin{gather}\label{226am109}
[\mathbf{A}_{\sigma}]_{\pi^{(i)}, \nu} =\prod_{k=1}^{i-1} [\mathbf{w}_k ]_{\pi^{(i)},\pi^{(i)}}\prod_{k=i}^{N-1}[\mathbf{w}_k]_{\pi^{(k)},\pi^{(k+1)}}.
\end{gather}
\end{Theorem}
Theorem \ref{312am109} includes the results in Theorems~\ref{1013pm104} and~\ref{302pm106}. This will be shown in Section~\ref{329am109}.
The following results implement the formula~(\ref{226am109}).
\begin{Proposition}\label{459pm106}
Let $\sigma=w_1\cdots w_{N-1}$ be expressed as in Theorem~{\rm \ref{700pm410}} and let $\mathbf{w}_k$
be the matrix corresponding to $w_k={T}_{k_1}(\beta_{k_1},\alpha_{k_1})\cdots {T}_{k_l}(\beta_{k_l},\alpha_{k_l})$. $($Hence, $k_1 = k$, $k_2 = k-1$ and so on.$)$ Then,
 \begin{gather}\label{407am106}
 [\mathbf{w}_k]_{\pi^{(i)},\pi^{(i)}} =
 \begin{cases}
 Q_{\beta_{k_1}\alpha_{k_1}}S_{\beta_{k_2}\alpha_{k_2}}\cdots S_{\beta_{k_l}\alpha_{k_l}} &\text{if $i= k+1$},\\
 S_{\beta_{k_1}\alpha_{k_1}}\cdots S_{\beta_{k_l}\alpha_{k_l}} &\text{if $i> k+1$}
 \end{cases}
 \end{gather}
and
 \begin{gather}\label{1234am108}
 [\mathbf{w}_k]_{\pi^{(k)},\pi^{(k+1)}} =pT_{\beta_{k_1}\alpha_{k_1}}S_{\beta_{k_2}\alpha_{k_2}}\cdots S_{\beta_{k_l}\alpha_{k_l}}.
 \end{gather}
\end{Proposition}
\begin{Remark}
The formulas for $\nu = 12\cdots2, 21\cdots 1$ and $1\cdots 12$ corresponding to (\ref{226am109}), (\ref{407am106}) and (\ref{1234am108}) can be obtained in a very similar way by using the techniques used for the case $\nu= 2\cdots 21$, so we omit them in this paper.
\end{Remark}
\begin{Remark}
The technique in \cite{Lee-2020} provided a method to find the explicit formulas of the transition probabilities but $[\mathbf{A}_{\sigma}]_{\pi,\nu}$ could be a sum of multiple terms in general, and it was not sure that it could be simplified further or not. One advantage of using the method of expressing permutations in this paper is that we can directly obtain the simplified forms of $[\mathbf{A}_{\sigma}]_{\pi,\nu}$ without performing factorization. This will improve the computational efficiency to find the transition probabilities.
\end{Remark}

\subsection{Organization of the paper}
This paper is organized as follows. In Section~\ref{main}, we provide special ways of expressing $\sigma$ as a product of simple transpositions for simplified forms of $[\mathbf{A}_{\sigma}]_{\pi, \nu}$ for $\nu = 2{\cdots}21, 12{\cdots}2, 21{\cdots}1, 1{\cdots}12$. The proof of Theorem~\ref{420am36} is given in Section \ref{755pm44}. In Section~\ref{300pm106}, we prove Theorem~\ref{1013pm104} and provide the formulas of $[\mathbf{A}_{\sigma}]_{\nu, \nu}$ for $\nu = 12\cdots 2, 21\cdots 1, 1\cdots 12$ in Propositions \ref{323am109},~\ref{1122pm104}, and~\ref{101149pm104}. The proof of Theorem~\ref{302pm106} is given in Section~\ref{326am109}, the proofs of Theorem~\ref{312am109} and Proposition~\ref{459pm106} are given in Section~\ref{329am109}. In Appendix~\ref{appendix-A}, we introduce an alternate approach to find $[\mathbf{A}_{\sigma}]_{\pi, \nu}$. In particular, the method introduced in Example~\ref{317pm127} is expected to be useful to write a computer code to find $[\mathbf{A}_{\sigma}]_{\pi, \nu}$. In Appendices~\ref{715pm44} and~\ref{7155pm44}, we provide all transition probabilities for $N=4$ when initial orders are $2221$ and ${1112}$ so that one can see what probability distribution can be obtained in closed forms.

\section[How to express sigma for simplified \protect{[A\_\{sigma\}]\_\{pi,nu\}}]{How to express $\boldsymbol{\sigma}$ for simplified $\boldsymbol{[\mathbf{A}_{\sigma}]_{\pi,\nu}}$}\label{main}

\subsection{Rearrangement of the columns and the rows}\label{454pm321}
In \cite{Lee-2020}, the columns and the rows of the $N^N \times N^N$ matrices are labelled by $1\cdots 1, \dots, N\cdots N$ in the lexicographical order. But, if we rearrange the columns and the rows by grouping their labels from the same multi-set $\mathcal{M}$, then $\mathbf{P}_Y(X;t)$, $\mathbf{A}_{\sigma}$ and $\mathbf{T}_i(\beta,\alpha)$ become block-diagonal. In each group, we list the columns and the rows in the lexicographical order. If we follow this procedure of rearranging the columns and the rows, then, for example, the form of the matrices~$\mathbf{P}_Y(X;t)$,~$\mathbf{A}_{\sigma}$ and $\mathbf{T}_i(\beta,\alpha)$ with $N=3$ looks as in Figure~\ref{fig:M2}. We denote the block corresponding to a multi-set~$\mathcal{M}$ in the $N^N \times N^N$ matrix~$\mathbf{A}$ by $\mathbf{A}^{\mathcal{M}}$. In other words, $\mathbf{A}^{\mathcal{M}}$~is a~sub-matrix of $\mathbf{A}$ obtained by taking the columns and the rows whose labels are the permutations of the multi-set~$\mathcal{M}$. Although the columns and the rows of~$\mathbf{A}^{\mathcal{M}}$ are labelled by permutations~$\pi$,~$\nu$ of $\mathcal{M}$ so that~$\big[\mathbf{A}^{\mathcal{M}}\big]_{\pi,\nu}$ is the $(\pi,\nu)^{\rm th}$ element of~$\mathbf{A}^{\mathcal{M}}$, we will sometimes write $\big[\mathbf{A}^{\mathcal{M}}\big]_{i,j}$ where $i,j = 1, \dots, L$ for the $(i,j)^{\rm th}$ element in the usual notation for matrix elements when necessary if there is not confusion. Here, $L$ is the total number of permutations of~$\mathcal{M}$.
 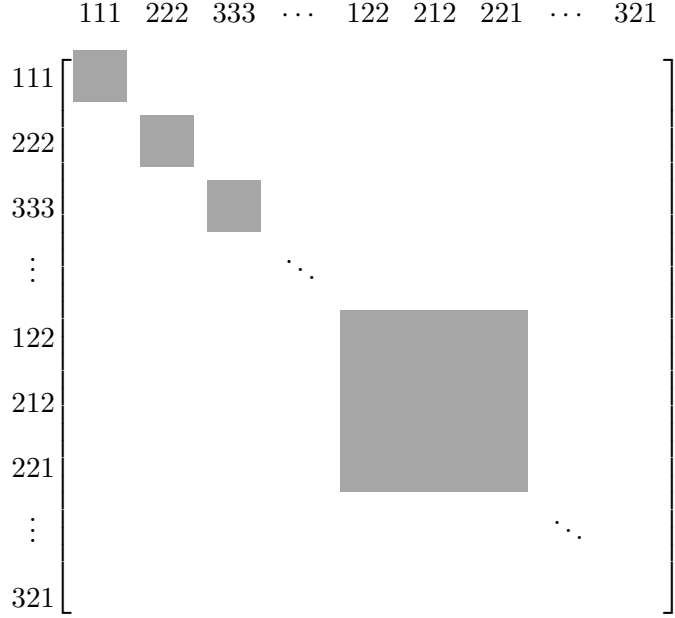
\begin{figure}[t]
 \centering \begin{tikzpicture}[
 >=latex,
 line width=1pt,
 Brace/.style={
 decorate,
 decoration={
 brace,
 raise=-7pt
 }
 },
 Matrix/.style={
 matrix of nodes,
 text height=2.5ex,
 text depth=0.75ex,
 text width=3.4ex,
 align=center,
 column sep=2pt,
 row sep=2pt,
 nodes in empty cells,
 },
 DA/.style={
 fill,
 opacity=0.7,
 inner sep=-2.0pt,
 line width=1pt,
 },
 DL/.style={
 left delimiter=[,
 right delimiter=],
 inner sep=-3pt,
 }
 ]

 \matrix[Matrix] at (0,0) (M){
 & 111 & 222 & 333 & $\cdots$ & 122 & 212 & 221& $\cdots$& 321\\
 111 & & & & & & & & &\\
 222 & & & & & & & & &\\
 333 & & & & & & & & &\\
 \vdots & & & & $\ddots$ & & & & &\\
 122 & & & & & & & & &\\
 212 & & & & & & & & &\\
 221 & & & & & & & & &\\
 \vdots & & & & & & & & $\ddots$&\\
 321 & & & & & & & & &\\
 };

 \begin{scope}[on background layer]
 \node[DL,fit=(M-2-2)(M-10-10)](subM-1){};
 \node[DA,gray,fit=(M-2-2)(M-2-2)](subM-2){};
 \node[DA,gray,fit=(M-3-3)(M-3-3)](subM-2){};
 \node[DA,gray,fit=(M-4-4)(M-4-4)](subM-2){};
 \node[DA,gray,fit=(M-6-6)(M-8-8)](subM-3){};

 \end{scope}

 \end{tikzpicture}
 \caption{The form of the matrix after rearranging the columns and the rows.} \label{fig:M2}
\end{figure}

\subsection{Proof of Theorem \ref{420am36}}\label{755pm44}
An expression of $\sigma$ obtained by Theorem~\ref{700pm410} is a reduced expression of $\sigma$, that is, no other expression of $\sigma$ is shorter than that. Rewriting each $w_i$ in $\sigma = w_1w_2\cdots w_{N-1}$ explicitly by simple transpositions, we obtain an expression of $\sigma$ in terms of simple transpositions
\begin{gather}\label{713pm321}
\sigma =w_1w_2\cdots w_{N-1} = T_{i_j}(\beta_j,\alpha_j)\cdots T_{i_2}(\beta_2,\alpha_2) T_{i_1}(\beta_1,\alpha_1)
\end{gather}
and
\begin{gather*}
\{(\beta_1,\alpha_1), (\beta_2,\alpha_2), \dots, (\beta_j,\alpha_j)\}
\end{gather*}
is the set of all inversions of $\sigma$. Using the notation introduced in Section~\ref{454pm321}, we note that
\begin{gather}\label{400pm320}
{\mathbf{T}}_{l}^{\mathcal{M}_N}(\beta,\alpha) = (S_{\beta\alpha}\mathbf{I}_{l-1} ) \oplus \left[
 \begin{matrix}
 P_{\beta\alpha} & pT_{\beta\alpha} \\
 qT_{\beta\alpha} & Q_{\beta\alpha}
 \end{matrix}
\right] \oplus (S_{\beta\alpha}\mathbf{I}_{N-l-1}),\qquad l=1,\dots, N-1,
\end{gather}
when $\mathcal{M}_N = [1,\underbrace{2,\dots,2}_{N-1}]$. The form of ${\mathbf{T}}_{l}^{\mathcal{M}_N}(\beta,\alpha)$ is as in Figure~\ref{fig:M11}.
\begin{figure}[t]
\centering
 \begin{tikzpicture}[
 >=latex,
 line width=1pt,
 Brace/.style={
 decorate,
 decoration={
 brace,
 raise=-7pt
 }
 },
 Matrix/.style={
 matrix of nodes,
 text height=2.05ex,
 text depth=0.75ex,
 text width=3ex,
 align=center,
 column sep=2pt,
 row sep=2pt,
 nodes in empty cells,
 },
 BOX/.style={
 fill,
 opacity=1,
 inner sep=-1.0pt,
 line width=0.8pt,
 draw,
 fill opacity=0,
 },
 DEL/.style={
 left delimiter=[,
 right delimiter=],
 inner sep=-3pt,
 }
 ]

 \matrix[Matrix] at (0,0) (M2){
 & 1 & $\cdots$ & $\cdots$ & \scriptsize$l$ & \scriptsize$l+1$ & $\cdots$ & $\cdots$&$N$\\
 1 & * & & & & & & & &\\
 \vdots & & $\ddots$ & & & & & & &\\
 \vdots & & & * & & & & & &\\
 \footnotesize$l$& & & & * & * & & & &\\
 \scriptsize$l+1$ & & & & *& *& & & &\\
 \vdots & & & & & & * & & &\\
 \vdots & & & & & & &$\ddots$ & &\\
 $N$ & & & & & & & &*&\\
 };

 \begin{scope}[on background layer]
 \node[DEL,fit=(M2-2-2)(M2-9-9)](subM-1){};
 \node[BOX,black,fit=(M2-5-5)(M2-6-6)](subM-2){};

 \end{scope}

 \end{tikzpicture}
 \caption{The form of the matrix $\mathbf{T}_l^{\mathcal{M}_N}$.} \label{fig:M11}
 \end{figure}
 Since we are interested in the last column of $\mathbf{A}_{\sigma}^{\mathcal{M}_N} = (\mathbf{w}_1\cdots \mathbf{w}_{N-1})^{\mathcal{M}_N}$, we investigate the last column of $(\mathbf{w}_{N-1})^{\mathcal{M}_N}$.

\begin{Lemma}\label{715pm310}
For $N \geq 2$ and $l=1,\dots, N-1$,
\begin{gather*}
\big[\mathbf{T}_{N-1}^{\mathcal{M}_N}(\beta_{N-1},\alpha_{N-1})\cdots \mathbf{T}_{l}^{\mathcal{M}_N}(\beta_{l},\alpha_{l}) \big]_{i,N} \\
\qquad{} =
\begin{cases}
0 & \text{if $i = 1,\dots, N-2$},\\
pT_{\beta_{N-1}\alpha_{N-1}}S_{\beta_{N-2}\alpha_{N-2}}\cdots S_{\beta_{l}\alpha_{l}}& \text{if $i = N-1$},\\
Q_{\beta_{N-1}\alpha_{N-1}}S_{\beta_{N-2}\alpha_{N-2}}\cdots S_{\beta_{l}\alpha_{l}}& \text{if $i = N$}.
\end{cases}
\end{gather*}
\end{Lemma}
\begin{proof}
If $N=2$, then
\begin{gather*}
\mathbf{T}_{1}^{\mathcal{M}_2}(\beta_{1},\alpha_{1}) = \left[
 \begin{matrix}
 P_{\beta_1\alpha_1} & pT_{\beta_1\alpha_1} \\
 qT_{\beta_1\alpha_1} & Q_{\beta_1\alpha_1}
 \end{matrix}
\right].
\end{gather*}
Suppose that $N>2$.
First, we observe that
\begin{gather}\label{333pm322}
\mathbf{T}^{\mathcal{M}_N}_l(\beta,\alpha) = \mathbf{T}^{\mathcal{M}_{N-1}}_{l}(\beta,\alpha)\oplus S_{\beta\alpha} \qquad \text{for $l = 1,\dots, N-2$}
\end{gather}
 and
\begin{gather*}
\mathbf{T}_{N-1}^{\mathcal{M}_N}(\beta,\alpha) = (S_{\beta\alpha}\mathbf{I}_{N-2} ) \oplus \left[
 \begin{matrix}
 P_{\beta\alpha} & pT_{\beta\alpha} \\
 qT_{\beta\alpha} & Q_{\beta\alpha}
 \end{matrix}
\right].
\end{gather*}
Hence,
\begin{gather*}
\mathbf{T}^{\mathcal{M}_N}_{N-1}(\beta_{N-1},\alpha_{N-1})\cdots \mathbf{T}^{\mathcal{M}_N}_l(\beta_{l},\alpha_{l}) \\
\qquad{} = \left( (S_{\beta_{N-1}\alpha_{N-1}}\mathbf{I}_{N-2} ) \oplus \left[
 \begin{matrix}
 P_{\beta_{N-1}\alpha_{N-1}} & pT_{\beta_{N-1}\alpha_{N-1}} \\
 qT_{\beta_{N-1}\alpha_{N-1}} & Q_{\beta_{N-1}\alpha_{N-1}}
 \end{matrix}
\right]\right) \\
\qquad\quad{} \times \big(\mathbf{T}^{\mathcal{M}_{N-1}}_{N-2}(\beta_{N-2},\alpha_{N-2}) \oplus S_{\beta_{N-2}\alpha_{N-2}}\big)\cdots \big(\mathbf{T}^{\mathcal{M}_{N-1}}_{l}(\beta_{l},\alpha_{l}) \oplus S_{\beta_{l}\alpha_{l}}\big) \\
\qquad {} = \left(\big(S_{\beta_{N-1}\alpha_{N-1}}\mathbf{I}_{N-2}\big) \oplus \left[
 \begin{matrix}
 P_{\beta_{N-1}\alpha_{N-1}} & pT_{\beta_{N-1}\alpha_{N-1}} \\
 qT_{\beta_{N-1}\alpha_{N-1}} & Q_{\beta_{N-1}\alpha_{N-1}}
 \end{matrix}
\right]\right) \\
\qquad\quad{} \times \big(\mathbf{T}^{\mathcal{M}_{N-1}}_{N-2}(\beta_{N-2},\alpha_{N-2}) \cdots \mathbf{T}^{\mathcal{M}_{N-1}}_{l}(\beta_{l},\alpha_{l})\big)\oplus \big(S_{\beta_{N-2}\alpha_{N-2}}\cdots S_{\beta_{l}\alpha_{l}}\big),
\end{gather*}
whose the last column is clearly
\begin{gather*}
\Bigg[ \underbrace{0~ \cdots ~ 0}_{N-2}~~ pT_{\beta_{N-1}\alpha_{N-1}}\prod_{i=l}^{N-2}S_{\beta_i\alpha_i} ~~ Q_{\beta_{N-1}\alpha_{N-1}}\prod_{i=l}^{N-2}S_{\beta_i\alpha_i} \Bigg]^{\rm T}.\tag*{\qed}
\end{gather*}\renewcommand{\qed}{}
\end{proof}

\begin{proof}[Proof of Theorem \ref{420am36}]
For all $N \geq 2$, if $\sigma$ is the identity permutation, then it is obvious that $[{\mathbf{A}}_{\sigma}]_{2\cdots 21, 2\cdots 21} = 1 $ and $ [{\mathbf{A}}_{\sigma}]_{\pi, 2\cdots 21} = 0 $ for $\pi \neq 2\cdots 21$ because $\mathbf{A}_{12\cdots N}$ is the identity matrix. Suppose that $\sigma$ is not the identity permutation. We prove by induction on $N$. If $N=2$, then $\mathbf{A}_{21} = \mathbf{R}_{21}$. Hence,
\begin{gather*}
[\mathbf{A}_{21}]_{21,21} = Q_{21}, \qquad
[\mathbf{A}_{21}]_{12,21} = pT_{21}.
\end{gather*}
 Suppose that the statement holds for $N-1$. Recall (\ref{333pm322}). If $w_{N-1} = 1$ in~(\ref{713pm321}), then
 \begin{align*}
\mathbf{A}_{\sigma}^{\mathcal{M}_N} &= \mathbf{T}_{i_j}^{\mathcal{M}_N}(\beta_j,\alpha_j)\cdots \mathbf{T}_{i_2}^{\mathcal{M}_N}(\beta_2,\alpha_2)\mathbf{T}_{i_1}^{\mathcal{M}_N}(\beta_1,\alpha_1) \\
&=\big(\mathbf{T}_{i_j}^{\mathcal{M}_{N-1}}(\beta_j,\alpha_j)\oplus S_{\beta_j\alpha_j}\big)\cdots \big(\mathbf{T}_{i_1}^{\mathcal{M}_{N-1}}(\beta_1,\alpha_1) \oplus S_{\beta_1\alpha_1}\big) \\
& = \big(\mathbf{T}_{i_j}^{\mathcal{M}_{N-1}}(\beta_j,\alpha_j) \cdots \mathbf{T}_{i_1}^{\mathcal{M}_{N-1}}(\beta_1,\alpha_1)\big) \oplus \prod_{i=1}^jS_{\beta_i\alpha_i}.
\end{align*}
Hence, the last column of $\mathbf{A}_{\sigma}^{\mathcal{M}_N}$ is
\begin{gather*}
\bigg[0~ \cdots~ f0~~ \prod_{i=1}^jS_{\beta_i\alpha_i} \bigg]^T.
\end{gather*}
If
\begin{gather*}
w_{N-1} = T_{i_m}(\beta_m,\alpha_m) \cdots T_{i_1}(\beta_1,\alpha_1)
\end{gather*}
in (\ref{713pm321}) so that $ T_{i_m}(\beta_m,\alpha_m) = T_{N-1}(\beta_{m},\alpha_{m})$, then
 \begin{gather*}
\big[\mathbf{A}_{\sigma}^{\mathcal{M}_N}\big]_{\pi, 2\cdots2 1} \\
\quad{}= \sum_{\nu}\big[\mathbf{T}_{i_j}^{\mathcal{M}_N}(\beta_j,\alpha_j)\cdots \mathbf{T}_{i_{m+1}}^{\mathcal{M}_N}(\beta_{m+1},\alpha_{m+1}) \big]_{\pi,\nu} \big[\mathbf{T}_{i_{m}}^{\mathcal{M}_N}(\beta_{m},\alpha_{m})\cdots\mathbf{T}_{i_1}^{\mathcal{M}_N}(\beta_1,\alpha_1)\big]_{\nu, 2\cdots 21}.
\end{gather*}
We have
\begin{gather*}
\mathbf{T}_{i_j}^{\mathcal{M}_N}(\beta_j,\alpha_j)\cdots \mathbf{T}_{i_{m+1}}^{\mathcal{M}_N}(\beta_{m+1},\alpha_{m+1}) \\
\qquad{} =\big(\mathbf{T}_{i_j}^{\mathcal{M}_{N-1}}(\beta_j,\alpha_j)\oplus S_{\beta_j\alpha_j}\big)\cdots \big(\mathbf{T}_{i_{m+1}}^{\mathcal{M}_{N-1}}(\beta_{m+1},\alpha_{m+1}) \oplus S_{\beta_{m+1}\alpha_{m+1}}\big) \\
\qquad{} = \big(\mathbf{T}_{i_j}^{\mathcal{M}_{N-1}}(\beta_j,\alpha_j) \cdots \mathbf{T}_{i_{m+1}}^{\mathcal{M}_{N-1}}(\beta_{m+1},\alpha_{m+1})\big) \oplus \prod_{i=m+1}^{j}S_{\beta_i\alpha_i},
\end{gather*}
whose the $i^{{\rm th}}$ row ($i=1,\dots, N-3$) is in the form of
 \begin{gather*}
\big[~\underbrace{*~\cdots ~*}_{N-2} ~ 0~~0~\big]
\end{gather*}
and the $(N-1)^{{\rm st}}$ and the $(N-2)^{{\rm nd}}$ rows are in the form of
 \begin{gather*}
\Bigg[\underbrace{*~\cdots ~*}_{N-2} ~~ \prod_{i=m+1}^{j}R_{\beta_i\alpha_i}~~0 \Bigg]
\end{gather*}
by the induction hypothesis for $N-1$. The $N^{{\rm th}}$ row is
 \begin{gather*}
\Bigg[\underbrace{0~\cdots ~0}_{N-1}~~\prod_{i=m+1}^{j}S_{\beta_i\alpha_i}\Bigg].
\end{gather*}
By Lemma \ref{715pm310}, the last column of $\mathbf{T}_{i_{m}}^{\mathcal{M}_N}(\beta_{m},\alpha_{m})\cdots\mathbf{T}_{i_1}^{\mathcal{M}_N}(\beta_1,\alpha_1)$ is
\begin{gather*}
\begin{aligned}
\Bigg[\underbrace{0~\cdots ~0}_{N-2}~~pT_{\beta_m\alpha_m}\prod_{i=1}^{m-1}S_{\beta_{i}\alpha_{i}} ~~ Q_{\beta_m\alpha_m}\prod_{i=1}^{m-1}S_{\beta_{i}\alpha_{i}}\Bigg]^T.
\end{aligned}
\end{gather*}
Hence, $\big[\mathbf{A}_{\sigma}^{\mathcal{M}_N}\big]_{\pi, 2\cdots2 1}$ is zero or in the form of $\prod_{i=1}^{j}R_{\beta_i\alpha_i}.$
\end{proof}

\subsection[Expressions of sigma for \protect{[A\_\{sigma\}]\_\{pi,12...2\}}]{Expressions of $\boldsymbol{\sigma}$ for $\boldsymbol{[\mathbf{A}_{\sigma}]_{\pi, 12\cdots 2}}$}\label{1034pm106}

If $\sigma \in S_N$ is expressed as in Theorem~\ref{700pm410}, then $[\mathbf{A}_{\sigma}]_{\pi, 12\cdots 2}$ may not be in a simplified form. In this section, we will find another way to express $\sigma$ as a product of simple transpositions for which $[\mathbf{A}_{\sigma}]_{\pi, 12\cdots 2}$ is directly written as a factorized form.
 According to \cite[Lemma~4.2]{Kassel}, every element $\sigma \in S_N$ can be written as a product of $T_1,\dots, T_{N-1}$ with $T_{N-1}$ appearing at most once. Similarly, we have the following result.
\begin{Lemma}\label{1140pm45}
Every element $\sigma \in S_N$ can be written as a product of $T_1,\dots, T_{N-1}$ with $T_1$ appearing at most once.
\end{Lemma}
Lemma \ref{1140pm45} can be proved by the mathematical induction in a similar way to the proof of Lemma~4.2 in~\cite{Kassel} by using the fact that the set of all bijections on $\{2,\dots, N\}$ is isomorphic to the set of all bijections on $\{1,\dots, N-1\}$. Proposition~\ref{442pm321}, similar to Theorem~\ref{700pm410}, provides a~method to express $\sigma$ as a product of simple transpositions with~$T_1$ appearing at most once, and we use this method for simplified forms of $[\mathbf{A}_{\sigma}]_{\pi, 12\cdots 2}$.
\begin{Proposition}\label{442pm321}
Consider the following subsets of the symmetric group~$S_N$:
\begin{gather*}
\begin{aligned}
&\Omega_{N-1} = \{1,T_1,T_2T_1,\dots, T_{N-1}\cdots T_2T_1\}, \\
&\Omega_{N-2} = \{1,T_2,T_3T_2, \dots, T_{N-1}\cdots T_2\}, \\
&\hspace{0.3cm} \vdots \\
&\Omega_{2} = \{1,T_{N-2}, T_{N-1}T_{N-2}\}, \\
&\Omega_{1} = \{1,T_{N-1}\}.
\end{aligned}
\end{gather*}
For any permutation $\sigma \in S_N$, there is a unique element
\begin{gather*}
(w_{N-1},\dots, w_{1}) \in \Omega_{N-1} \times \Omega_{N-2} \times \cdots \times \Omega_{1}
\end{gather*}
such that $\sigma = w_{N-1}w_{N-2}\cdots w_{1}$.
\end{Proposition}
\begin{proof}
We prove this by induction on $N$. The statement is obvious for $N=2$. Suppose that the statement holds for $N-1$. It suffices to prove for $\sigma \in S_N$ represented by $T_1,\dots,T_{N-1}$ in which $T_1$ appears exactly once by Lemma~\ref{1140pm45}. (If $T_1$ does not appear in the representation of $\sigma$, then $\sigma(1)=1$ and there exist $w_i \in \Omega_i,~i=1,\dots, N-2$ such that $\sigma$ is uniquely written as $\sigma = w_{N-2}\cdots w_{1}$ by the induction hypothesis for $N-1$.)
Let $\sigma =\sigma' T_1 \sigma''$ where $\sigma'$ and $\sigma''$ are \textit{words} consisting of $T_2,\dots, T_{N-1}$. By the induction hypothesis, there exist $w'_i \in \Omega_i,~i=1,\dots, N-2$ such that $\sigma'$ is uniquely written as $\sigma' = w'_{N-2}\cdots w'_{1}$ so that $\sigma = w'_{N-2}\cdots w'_{1}T_1\sigma''$. By~(\ref{952pm411}), $T_1$~commutes with $w'_1,\dots,w'_{N-3}$, so
\begin{gather*}
\sigma = w'_{N-2}T_1w'_{N-3}\cdots w'_{1}\sigma''.
\end{gather*}
Since $w'_{N-3}\cdots w'_{1}\sigma''$ is a \textit{word} consisting of $T_2,\dots, T_{N-1}$, it may be written uniquely as \linebreak $w''_{N-2}\cdots w''_{1}$ for some $w''_i \in \Omega_i$ by the induction hypothesis. It is clear that $ w'_{N-2}T_1\in \Omega_{N-1}$.
\end{proof}

$[\mathbf{A}_{\sigma}]_{\pi, 12\cdots 2}$ implies the first column of $\mathbf{A}_{\sigma}^{\mathcal{M}_N} = (\mathbf{w}_{N-1}\cdots \mathbf{w}_{1})^{\mathcal{M}_N}$. Let us find the form of the first column of $(\mathbf{w}_{N-1})^{\mathcal{M}_N}$ (The reason why we consider this will be clearer soon). Note that non-identity $w_{N-1} \in \Omega_{N-1}$ is in the form of $T_lT_{l-1}\cdots T_1$ for some $l$.
\begin{Lemma}\label{1048pm104}
For $N \geq 2$ and $l=1,\dots, N-1$,
\begin{gather*}
\big[\mathbf{T}_{l}^{\mathcal{M}_N}(\beta_l,\alpha_l)\cdots \mathbf{T}_{1}^{\mathcal{M}_N}(\beta_{1},\alpha_{1}) \big]_{i,1}
\end{gather*}
is in the form of $\prod_{k=1}^lR_{\beta_k\alpha_k}$ where $R_{\beta\alpha}$ is one of $S_{\beta\alpha}$, $P_{\beta\alpha}$, $Q_{\beta\alpha}$, $pT_{\beta\alpha}$ or $qT_{\beta\alpha}$ if $1\leq i \leq l+1$, and is zero if $l+1<i \leq N$.
\end{Lemma}
\begin{proof}
We prove by induction on $N$. If $N=2$,
\begin{gather*}
\mathbf{T}_{1}^{\mathcal{M}_2}(\beta_{1},\alpha_{1}) = \left[
 \begin{matrix}
 P_{\beta_1\alpha_1} & pT_{\beta_1\alpha_1} \\
 qT_{\beta_1\alpha_1} & Q_{\beta_1\alpha_1}
 \end{matrix}
\right].
\end{gather*}
Suppose that the statement is true for $N-1$. If $l=1,\dots, N-2$, the induction hypothesis and~(\ref{333pm322}) imply that
\begin{gather*}
\big[ \mathbf{T}_{l}^{\mathcal{M}_N}(\beta_l,\alpha_l)\cdots \mathbf{T}_{1}^{\mathcal{M}_N}(\beta_{1},\alpha_{1}) \big]_{i,1}
\end{gather*}
is in the form of $\prod_{k=1}^lR_{\beta_k\alpha_k}$ if $1\leq i \leq l+1$, and is zero if $l+1<i \leq N$.
Let $l=N-1$. Then,
\begin{gather*}
\big[\mathbf{T}_{N-2}^{\mathcal{M}_N}(\beta_{N-2},\alpha_{N-2})\cdots \mathbf{T}_{1}^{\mathcal{M}_N}(\beta_{1},\alpha_{1}) \big]_{i,1}
\end{gather*}
is in the form of $\prod_{k=1}^{N-2}R_{\beta_k\alpha_k}$ for $i=1,\dots, N-1$ and is zero for $i=N$ by the induction hypothesis. Noting that
\begin{gather*}
\mathbf{T}_{N-1}^{\mathcal{M}_N}(\beta,\alpha) =
 \big(S_{\beta\alpha}\mathbf{I}_{N-2}\big) \oplus \left[ \begin{matrix}
 P_{\beta\alpha} & pT_{\beta\alpha} \\
 qT_{\beta\alpha} & Q_{\beta\alpha}
 \end{matrix}
\right]
\end{gather*}
and computing
\begin{gather}\label{1124pm107}
\big[ \mathbf{T}_{N-1}^{\mathcal{M}_N}(\beta_{N-1},\alpha_{N-1})\mathbf{T}_{N-2}^{\mathcal{M}_N}(\beta_{N-2},\alpha_{N-2})\cdots \mathbf{T}_{1}^{\mathcal{M}_N}(\beta_{1},\alpha_{1}) \big]_{i,1}
\end{gather}
directly, we can show that~(\ref{1124pm107}) is in the form of $\prod_{k=1}^{N-1}R_{\beta_k\alpha_k}$ for all $i$.
\end{proof}

Rewriting each $w_i$ in $\sigma = w_{N-1}\cdots w_1$ in Proposition~\ref{442pm321} explicitly by simple transpositions, we obtain an expression of $\sigma$ in terms of simple transpositions
\begin{gather*}
\sigma =w_{N-1}\cdots w_2w_{1} = T_{i_j}(\beta_j,\alpha_j)\cdots T_{i_2}(\beta_2,\alpha_2)T_{i_1}(\beta_1,\alpha_1)
\end{gather*}
and
\begin{gather*}
\{(\beta_1,\alpha_1),(\beta_2,\alpha_2), \dots, (\beta_j,\alpha_j)\}
\end{gather*}
is the set of all inversions in $\sigma$.
\begin{Theorem}\label{420am366}
Let $\sigma = w_{N-1}\cdots w_1 $ be an expression as in Proposition~{\rm \ref{442pm321}} and let $\mathbf{A}_{\sigma}$
be the matrix corresponding $\sigma$. Then, for all $N \geq 2$,
\begin{itemize}\itemsep=0pt
 \item [$(a)$] $ \big[{\mathbf{A}}_{12\cdots N}\big]_{12\cdots 2, 12\cdots 2} = 1 $ and $ \big[{\mathbf{A}}_{12\cdots N}\big]_{\pi, 12\cdots 2} = 0 $ if $\pi \neq 12\cdots 2$.
 \item [$(b)$] If $\sigma \neq 12\cdots N$, then $ \big[{\mathbf{A}}_{\sigma}\big]_{\pi, 12\cdots 2}$ is zero or written as
 \begin{gather}\label{428am366}
 \big[{\mathbf{A}}_{\sigma}\big]_{\pi, 12\cdots 2} = \prod_{(\beta,\alpha)}R_{\beta\alpha}
\end{gather}
where $R_{\beta\alpha}$ is one of $ S_{\beta\alpha}$, $P_{\beta\alpha}$, $Q_{\beta\alpha}$, $pT_{\beta\alpha}$, $qT_{\beta\alpha}$ and the product in~\eqref{428am366} is taken over all inversions $(\beta,\alpha)$ in~$\sigma$.
\end{itemize}
\end{Theorem}
\begin{proof}
Let
\begin{gather*}
\mathbf{A}_{\sigma} = \mathbf{w}_{N-1}\cdots \mathbf{w}_1
 =\underbrace{\mathbf{T}_{i_j}(\beta_j,\alpha_j)\cdots \mathbf{T}_{i_{m+1}}(\beta_{m+1},\alpha_{m+1})}_{=\mathbf{w}_{N-1}}~\underbrace{\mathbf{T}_{i_m}(\beta_m,\alpha_m)\cdots \mathbf{T}_{i_1}(\beta_1,\alpha_1)}_{=\mathbf{w}_{N-2}\cdots \mathbf{w}_1}.
\end{gather*}
For all $N \geq 2$, if $\sigma$ is the identity permutation, then it is obvious that $[{\mathbf{A}}_{\sigma}]_{12\cdots 2, 12\cdots 2} = 1 $ and $ [{\mathbf{A}}_{\sigma}]_{\pi, 12\cdots 2} = 0 $ for $\pi \neq 12\cdots 2$ because $\mathbf{A}_{12\cdots N}$ is the identity matrix. Suppose that $\sigma$ is not the identity permutation. Since there is no $T_1$ in $w_{N-2}\cdots w_1$,
\begin{gather*}
(\mathbf{w}_{N-2}\cdots \mathbf{w}_1)^{\mathcal{M}_N} = \prod_{i=1}^mS_{\beta_i\alpha_i} \oplus \mathbf{G}_{N-1}
\end{gather*}
for some $(N-1) \times (N-1)$ matrix $\mathbf{G}_{N-1}$ by (\ref{400pm320}). By Lemma \ref{1048pm104}, $\big[(\mathbf{w}_{N-1})^{\mathcal{M}_N}\big]_{l,1}$ is zero if $i_j+1 < l \leq N$, and is in the form of $\prod_{k=m+1}^{j}R_{\beta_k\alpha_k}$ if $1\leq l \leq i_j+1$. Hence,
\begin{gather*}
\big[(\mathbf{w}_{N-1})^{\mathcal{M}_N}(\mathbf{w}_{N-2}\cdots \mathbf{w}_1)^{\mathcal{M}_N} \big]_{l,1} = \prod_{i=1}^mS_{\beta_i\alpha_i}\prod_{k=m+1}^{j}R_{\beta_k\alpha_k}
\end{gather*}
for some $R_{\beta_k\alpha_k}$ if $1\leq l \leq i_j+1$, and is zero $i_j+1 < l \leq N$.
\end{proof}

\subsection[Expressions of sigma for \protect{[A\_\{sigma\}]\_\{pi,21...1\}} and \protect{[A\_\{sigma\}]\_\{pi,1...12\}}]{Expressions of $\boldsymbol{\sigma}$ for $\boldsymbol{[\mathbf{A}_{\sigma}]_{\pi, 21\cdots 1}}$ and $\boldsymbol{[\mathbf{A}_{\sigma}]_{\pi, 1\cdots 12}}$}\label{1035pm106}

In order to find $[\mathbf{A}_{\sigma}]_{\pi, 21\cdots 1}$ in factorized forms, $\sigma$ should be expressed as follows.
\begin{Proposition}\label{514pm323}
Consider the following subsets of the symmetric group $S_N$:
\begin{gather*}
\begin{aligned}
&\Gamma_{N-1} = \{1,T_{N-1}\}, \\
&\Gamma_{N-2} = \{1,T_{N-2},T_{N-2}T_{N-1}\}, \\
&\Gamma_{N-3} = \{1,T_{N-3},T_{N-3}T_{N-2},T_{N-3}T_{N-2}T_{N-1}\}, \\
&\hspace{0.3cm} \vdots \\
&\Gamma_{1} = \{1,T_{1},T_{1}T_{2},\dots, T_{1}\cdots T_{N-2}T_{N-1}\}.
\end{aligned}
\end{gather*}
For any permutation $\sigma \in S_N$, there is a unique element
\begin{gather*}
(w_{N-1},\dots, w_2,w_{1}) \in \Gamma_{N-1} \times \cdots \times \Gamma_2 \times \Gamma_{1}
\end{gather*}
such that $\sigma = w_{N-1}\cdots w_2w_{1}$.
\end{Proposition}
\begin{proof}
We prove this by induction on $N$. The statement is obvious for $N=2$. Suppose that the statement holds for $N-1$. It suffices to prove for $\sigma \in S_N$ represented by $T_1,\dots,T_{N-1}$ in which $T_1$ appears exactly once by Lemma~\ref{1140pm45}. If $T_1$ does not appear in the representation of $\sigma$, then $\sigma(1)=1$ and $\sigma$ is uniquely written $\sigma = w_{N-1}\cdots w_{2}$ by the induction hypothesis for $N-1$ because the set of all permutations on $\{1,\dots, N-1\}$ is isomorphic to the set of all permutations on $\{2,\dots, N\}$.
Let $\sigma =\sigma' T_1 \sigma''$ where $\sigma'$ and $\sigma''$ are \textit{words} consisting of $T_2,\dots, T_{N-1}$. By the induction hypothesis, there exist $w''_i \in \Gamma_i,~i=2,\dots, N-1$ such that $\sigma''$ is uniquely written as $\sigma'' = w''_{N-1}\cdots w''_{2}$ so that $\sigma = \sigma'T_1w''_{N-1}\cdots w''_{2}$. By (\ref{952pm411}), $T_1$ commutes with $w''_3,\dots,w''_{N-1}$, so
\begin{gather*}
\sigma = \sigma'w''_{N-1}\cdots w''_{3}T_1w''_2.
\end{gather*}
Since $\sigma'w''_{N-1}\cdots w''_{3}$ is a \textit{word} consisting of $T_2,\dots, T_{N-1}$, it may be written uniquely as \linebreak $w'_{N-1}\cdots w'_{2}$ for some $w'_i \in \Gamma_i$, $i=2,\dots, N-1$ by the induction hypothesis. It is clear that $ T_1w''_2\in \Gamma_{1}$.
\end{proof}

\begin{Corollary}\label{1112pm44}
Let $\sigma =w_{N-1}\cdots w_1$ be an expression as in Proposition~{\rm \ref{514pm323}} and let~$\mathbf{A}_{\sigma}$ be the matrix corresponding to $\sigma$.
 Then, for all $N \geq 2$,
\begin{itemize}\itemsep=0pt
 \item [$(a)$] $ [{\mathbf{A}}_{12\cdots N}]_{21\cdots 1, 21\cdots 1} = 1 $ and $ [{\mathbf{A}}_{12\cdots N}]_{\pi, 21\cdots 1} = 0 $ if $\pi \neq 21\cdots 1$.
 \item [$(b)$] If $\sigma \neq 12\cdots N$, then $[{\mathbf{A}}_{\sigma}]_{\pi, 21\cdots 1}$ is zero or written as
 \begin{gather}\label{428am3666}
 [{\mathbf{A}}_{\sigma}]_{\pi, 21\cdots 1} = \prod_{(\beta,\alpha)}R_{\beta\alpha},
\end{gather}
where $R_{\beta\alpha}$ is one of $ S_{\beta\alpha}$, $P_{\beta\alpha}$, $Q_{\beta\alpha}$, $pT_{\beta\alpha}$, $qT_{\beta\alpha}$ and the product in~\eqref{428am3666} is taken over all inversions $(\beta,\alpha)$ in $\sigma$.
\end{itemize}
\end{Corollary}
\begin{Remark}
Let $\mathcal{N}_N $ be the multi-set $[\underbrace{1,\dots,1}_{N-1},2]$. Then, we observe that
\begin{gather}\label{1055pm44}
{\mathbf{T}}_{l}^{\mathcal{N}_N}(\beta,\alpha) = (S_{\beta\alpha}\mathbf{I}_{N-l-1}) \oplus \left[
 \begin{matrix}
 P_{\beta\alpha} & pT_{\beta\alpha} \\
 qT_{\beta\alpha} & Q_{\beta\alpha} \\
 \end{matrix}
\right] \oplus (S_{\beta\alpha}\mathbf{I}_{l-1}),\qquad l=1,\dots, N-1,
\end{gather}
and $ \mathbf{T}_{l}^{\mathcal{N}_N}(\beta,\alpha) = \mathbf{T}_{N-l}^{\mathcal{M}_N}(\beta,\alpha)$. Using these properties, Corollary \ref{1112pm44} can be proved essentially in the same way as the proof Theorem~\ref{420am36}.
\end{Remark}
Simplified forms of $[\mathbf{A}_{\sigma}]_{\pi, 1\cdots 12}$ are obtained via the following method of expressing $\sigma$.
\begin{Proposition}\label{1044pm44}
Consider the following subsets of the symmetric group $S_N$:
\begin{gather*}
\begin{aligned}
&\Xi_{N-1} = \{1,T_{N-1},T_{N-2}T_{N-1},\dots, T_{1}\cdots T_{N-1}\}, \\
&\Xi_{N-2} = \{1,T_{N-2},T_{N-3}T_{N-2}, \dots, T_{1}\cdots T_{N-2}\}, \\
&\hspace{0.3cm} \vdots \\
&\Xi_{2} = \{1,T_{2}, T_{1}T_{2}\}, \\
&\Xi_{1} = \{1,T_{1}\}.
\end{aligned}
\end{gather*}
For any permutation $\sigma \in S_N$, there is a unique element
\begin{gather*}
(w_{N-1},\dots, w_{1}) \in \Xi_{N-1} \times \cdots \times \Xi_{1}
\end{gather*}
such that $\sigma = w_{N-1}\cdots w_{1}$.
\end{Proposition}
\begin{proof}
We prove this by induction on $N$. The statement is obvious for $N=2$. Suppose that the statement holds for $N-1$. It suffices to prove for $\sigma \in S_N$ represented by $T_1,\dots,T_{N-1}$ in which $T_{N-1}$ appears exactly once by \cite[Lemma~4.2]{Kassel}. If $T_{N-1}$ does not appear in the representation of $\sigma$, then $\sigma(N)=N$ and $\sigma$ is uniquely written $\sigma = w_{N-2}\cdots w_{1}$ by the induction hypothesis for $N-1$.
Let $\sigma =\sigma' T_{N-1} \sigma''$ where $\sigma'$ and $\sigma''$ are \textit{words} consisting of $T_1,\dots, T_{N-2}$. By the induction hypothesis, there exist $w'_i \in \Xi_i,~i=1,\dots, N-2$ such that $\sigma'$ is uniquely written as $\sigma' = w'_{N-2}\cdots w'_{1}$ so that $\sigma = w'_{N-2}\cdots w'_{1}T_{N-1}\sigma''$. By (\ref{952pm411}), $T_{N-1}$ commutes with $w'_1,\dots,w'_{N-3}$, so
\begin{gather*}
\sigma = w'_{N-2}T_{N-1}w'_{N-3}\cdots w'_{1}\sigma''.
\end{gather*}
Since $w'_{N-3}\cdots w'_{1}\sigma''$ is a \textit{word} consisting of $T_1,\dots, T_{N-2}$, it may be written as $w''_{N-2}\cdots w''_{1}$ for some $w''_i \in \Xi_i$ by the induction hypothesis. It is clear that $ w'_{N-2}T_{N-1}\in \Xi_{N-1}$.
\end{proof}

\begin{Corollary}\label{1111pm44}
Let $\sigma=w_{N-1}\cdots w_1$ be an expressed as in Proposition~{\rm \ref{1044pm44}} and let $\mathbf{A}_{\sigma}$ be the matrix corresponding to $\sigma$.
Then, for all $N \geq 2$,
\begin{itemize}\itemsep=0pt
 \item [$(a)$] $ \big[{\mathbf{A}}_{12\cdots N}\big]_{1\cdots 12, 1\cdots 12} = 1 $ and $ \big[{\mathbf{A}}_{12\cdots N}\big]_{\pi, 1\cdots 12} = 0 $ if $\pi \neq 1\cdots 12$.
 \item [$(b)$] If $\sigma \neq 12\cdots N$, then $ \big[{\mathbf{A}}_{\sigma}\big]_{\pi, 1\cdots 12}$ is zero or written as
 \begin{gather}\label{428am36666}
 \big[{\mathbf{A}}_{\sigma}\big]_{\pi, 1\cdots 12} = \prod_{(\beta,\alpha)}R_{\beta\alpha}
\end{gather}
where $R_{\beta\alpha}$ is one of $S_{\beta\alpha}$, $P_{\beta\alpha}$, $Q_{\beta\alpha}$, $pT_{\beta\alpha}$, $qT_{\beta\alpha}$ and the product in~\eqref{428am36666} is taken over all inversions $(\beta,\alpha)$ in~$\sigma$.
\end{itemize}
\end{Corollary}
\begin{Remark}
 Corollary~\ref{1111pm44} can be proved essentially in the same way as the proof Theorem~\ref{420am366} by using $\mathbf{T}_{l}^{\mathcal{N}_N}(\beta,\alpha) = \mathbf{T}_{N-l}^{\mathcal{M}_N}(\beta,\alpha)$.
\end{Remark}

\section[Simplified forms of \protect{[A\_\{sigma\}]\_\{pi,nu\}}]{Simplified forms of $\boldsymbol{[\mathbf{A}_{\sigma}]_{\pi,\nu}}$}\label{920am49}

In the previous section, we discussed the existence of the factorized forms of $[\mathbf{A}_{\sigma}]_{\pi,\nu}$. In this section, we provide the explicit forms of $[\mathbf{A}_{\sigma}]_{\pi,\nu}$.
\begin{Remark}
In this section, sometimes, we will simply write any $N \times N$ sub-matrix $\mathbf{A}^{\mathcal{M}_N}$ of~$\mathbf{A}$ as just $\mathbf{A}$ for notational convenience. Also, recall that we interchangeably use the notations~$[\mathbf{A}]_{i,j}$ where $i,j=1,\dots, N$ and $[\mathbf{A}]_{\pi,\nu}$ where $\pi$ and $\nu$ are permutations of a multi-set. For example, $[\mathbf{A}]_{1,N} = [\mathbf{A}]_{12\cdots 2, 2\cdots 21}$.
\end{Remark}

\subsection[\protect{[A\_\{sigma\}]\_\{nu,nu\}}]{$\boldsymbol{[\mathbf{A}_{\sigma}]_{\nu,\nu}}$}\label{300pm106}

Theorem \ref{1013pm104} states that $[\mathbf{A}_{\sigma}]_{2\cdots 21, 2\cdots 21}$ is nonzero for each $\sigma$ and expressed as a product over all inversions in~$\sigma$.
\begin{proof}[Proof of Theorem~\ref{1013pm104}]
First, recall that $T_{N-1}$ appears at most once in any expression of $\sigma$ given by Theorem \ref{700pm410}. Suppose that $i_m = N-1$ for some integer $1\leq m \leq j$ so that
\begin{gather*}
 {T}_{i_m}(\beta_m,\alpha_m)\cdots{T}_{i_1}(\beta_1,\alpha_1)= {T}_{N-1}(\beta_m,\alpha_m)\cdots{T}_{N-m}(\beta_1,\alpha_1)=w_{N-1} \in \Sigma_{N-1}.
\end{gather*}
By Lemma \ref{715pm310} and recalling the form of (\ref{400pm320}),
\begin{align*}
[\mathbf{T}_{i_m}(\beta_m,\alpha_m)\cdots\mathbf{T}_{i_1}(\beta_1,\alpha_1) ]_{N,N} & = Q_{\beta_m\alpha_m}S_{\beta_{m-1}\alpha_{m-1}}\cdots S_{\beta_1\alpha_1}\\
&= [\mathbf{T}_{N-1}(\beta_m,\alpha_m) ]_{N,N} \cdots [\mathbf{T}_{N-m}(\beta_1,\alpha_1) ]_{N,N}.
\end{align*}
Since $\mathbf{T}_{i_j}, \dots, \mathbf{T}_{i_{m+1}}$ are not from $w_{N-1}$, the integers $i_j,\dots, i_{m+1}$ are not equal to $N-1$. Hence, each of $\mathbf{T}_{i_j}, \dots, \mathbf{T}_{i_{m+1}}$ is in the form of $\mathbf{G}_{N-1} \oplus S_{\beta\alpha}$ where $\mathbf{G}_{N-1}$ is an $(N-1) \times (N-1)$ matrix (recall the form of (\ref{400pm320})). Hence,
\begin{gather*}
[\mathbf{T}_{i_j} \cdots \mathbf{T}_{i_{m+1}}]_{N,N} = [\mathbf{T}_{i_j}]_{N,N} \cdots [\mathbf{T}_{i_{m+1}}]_{N,N} = S_{\beta_{j}\alpha_{j}}\cdots S_{\beta_{m+1}\alpha_{m+1}}
\end{gather*}
and
\begin{gather*}
[\mathbf{A}_{\sigma}]_{N,N}= \big[\mathbf{T}_{i_j} \cdots \mathbf{T}_{i_1}\big]_{N,N} = S_{\beta_{j}\alpha_{j}}\cdots S_{\beta_{m+1}\alpha_{m+1}} Q_{\beta_m\alpha_m}S_{\beta_{m-1}\alpha_{m-1}}\cdots S_{\beta_1\alpha_1}.
\end{gather*}
If $i_1,\dots,i_j \neq N-1$, then each of $\mathbf{T}_{i_j}, \dots, \mathbf{T}_{i_1}$ is in the form of $\mathbf{G}_{N-1} \oplus S_{\beta\alpha}$. Hence,
\begin{gather*}
[\mathbf{A}_{\sigma}]_{N,N}= [\mathbf{T}_{i_j} \cdots \mathbf{T}_{i_1}]_{N,N} = S_{\beta_{j}\alpha_{j}}\cdots S_{\beta_1\alpha_1}.\tag*{\qed}
\end{gather*}\renewcommand{\qed}{}
\end{proof}

The formulas for other cases of $\nu$ that are considered in this paper can be similarly obtained by using the same techniques, so we provide their formulas without proofs.
\begin{Proposition}\label{323am109}
Let $\nu=12\cdots 2$ and
\begin{gather*}
\sigma = w_{N-1}\cdots w_{1} = {T}_{i_j}(\beta_j,\alpha_j)\cdots {T}_{i_1}(\beta_1,\alpha_1)
\end{gather*}
be given as in Proposition~{\rm \ref{442pm321}}. Let
\begin{gather*}
\mathbf{A}_{\sigma}=\mathbf{T}_{i_j}(\beta_j,\alpha_j)\cdots \mathbf{T}_{i_1}(\beta_1,\alpha_1)
\end{gather*}
 be the matrix corresponding to $\sigma$. Then,
\begin{align*}
 [\mathbf{A}_{\sigma}]_{\nu, \nu} &= [\mathbf{T}_{i_j} ]_{\nu, \nu} \cdots [\mathbf{T}_{i_m} ]_{\nu, \nu} \cdots [\mathbf{T}_{i_1} ]_{\nu, \nu} \\
& = \begin{cases}
S_{\beta_{j}\alpha_{j}}\cdots S_{\beta_{m+1}\alpha_{m+1}} P_{\beta_m\alpha_m}S_{\beta_{m-1}\alpha_{m-1}}\cdots S_{\beta_1\alpha_1}&\text{if~$i_m = 1$},\\
S_{\beta_{j}\alpha_{j}}\cdots S_{\beta_1\alpha_1}&\text{if~$i_i,\dots, i_j \neq 1$}.
 \end{cases}
\end{align*}
 $($Note that $T_1$ appears at most once in any expression of $\sigma$ in Proposition~{\rm \ref{442pm321}}.$)$
\end{Proposition}

\begin{Proposition}\label{1122pm104}
Let $\nu=21\cdots 1$ and
\begin{gather*}
\sigma = w_{N-1}\cdots w_{1} = {T}_{i_j}(\beta_j,\alpha_j)\cdots {T}_{i_1}(\beta_1,\alpha_1)
\end{gather*}
be given as in Proposition~{\rm \ref{514pm323}}. Let
\begin{gather*}
\mathbf{A}_{\sigma}=\mathbf{T}_{i_j}(\beta_j,\alpha_j)\cdots \mathbf{T}_{i_1}(\beta_1,\alpha_1)
\end{gather*}
 be the matrix corresponding to $\sigma$. Then,
\begin{align*}
 [\mathbf{A}_{\sigma}]_{\nu, \nu} &= [\mathbf{T}_{i_j}]_{\nu, \nu} \cdots [\mathbf{T}_{i_m}]_{\nu, \nu} \cdots [\mathbf{T}_{i_1}]_{\nu, \nu} \\
& = \begin{cases}
S_{\beta_{j}\alpha_{j}}\cdots S_{\beta_{m+1}\alpha_{m+1}} Q_{\beta_m\alpha_m}S_{\beta_{m-1}\alpha_{m-1}}\cdots S_{\beta_1\alpha_1}&\text{if~$i_m = 1$},\\
S_{\beta_{j}\alpha_{j}}\cdots S_{\beta_1\alpha_1}&\text{if~$i_i,\dots, i_j \neq 1$}.
 \end{cases}
\end{align*}
 $($Note that $T_1$ appears at most once in any expression of $\sigma$ in Proposition~{\rm \ref{514pm323}}.$)$
\end{Proposition}

\begin{Proposition}\label{101149pm104}
Let $\nu=1\cdots 12$ and
\begin{gather*}
\sigma = w_{N-1}\cdots w_{1} = {T}_{i_j}(\beta_j,\alpha_j)\cdots {T}_{i_1}(\beta_1,\alpha_1)
\end{gather*}
be given as in Proposition~{\rm \ref{1044pm44}}. Let
\begin{gather*}
\mathbf{A}_{\sigma}=\mathbf{T}_{i_j}(\beta_j,\alpha_j)\cdots \mathbf{T}_{i_1}(\beta_1,\alpha_1)
\end{gather*}
 be the matrix corresponding to $\sigma$. Then,
\begin{align*}
 [\mathbf{A}_{\sigma}]_{\nu, \nu} &= [\mathbf{T}_{i_j} ]_{\nu, \nu} \cdots [\mathbf{T}_{i_m} ]_{\nu, \nu} \cdots [\mathbf{T}_{i_1} ]_{\nu, \nu} \\
& = \begin{cases}
S_{\beta_{j}\alpha_{j}}\cdots S_{\beta_{m+1}\alpha_{m+1}} P_{\beta_m\alpha_m}S_{\beta_{m-1}\alpha_{m-1}}\cdots S_{\beta_1\alpha_1}&\text{if~$i_m = N-1$},\\
S_{\beta_{j}\alpha_{j}}\cdots S_{\beta_1\alpha_1}&\text{if~$i_i,\dots, i_j \neq 1$}.
 \end{cases}
\end{align*}
 $($Note that $T_{N-1}$ appears at most once in any expression of $\sigma$ in Proposition~{\rm \ref{1044pm44}}.$)$
\end{Proposition}

\subsection[Case \protect{[A\_\{sigma\}]\_\{pi,nu\}=0} where pi not= nu = 2...21]{Case $\boldsymbol{[\mathbf{A}_{\sigma}]_{\pi,\nu}=0}$ where $\boldsymbol{\pi\neq \nu = 2\cdots 21}$}\label{326am109}

If $\pi\neq \nu= 2\cdots 21$, then $[\mathbf{A}_{\sigma}]_{\pi,\nu}$ where $\pi\neq \nu$ is either zero or a product of some factors where the product is taken over all inversions in $\sigma$. First, let us investigate when it can be zero.
Let $\sigma= w_1\cdots w_{N-1}$ be expressed as in Theorem \ref{700pm410} and let
\begin{gather*}
\mathbf{w}_m = \mathbf{T}_m(\beta_m,\alpha_m)\mathbf{T}_{m-1}(\beta_{m-1},\alpha_{m-1}) \cdots \mathbf{T}_l(\beta_{l},\alpha_{l})
\end{gather*}
 if $w_m = T_m(\beta_m,\alpha_m)T_{m-1}(\beta_{m-1},\alpha_{m-1})\cdots T_l(\beta_{l},\alpha_{l})$, and let $\mathbf{w}_m$ be the identity matrix if \mbox{$w_m=1$}. Recalling the form of the matrix in (\ref{400pm320}), we see that $\mathbf{w}_m$ is in the form of $\mathbf{G}_{m+1}\oplus \mathbf{D}_{N-m-1}$ where $\mathbf{G}_{m+1}$ is an $(m+1) \times (m+1)$ matrix and $\mathbf{D}_{N-m-1}$ is an $(N-m-1)\times (N-m-1)$ diagonal matrix.
An immediate consequence of Lemma~\ref{715pm310} is as follows:
\begin{Lemma}\label{213am106}
The $(m+1)^{{\rm st}}$ column of $\mathbf{w}_m \neq 1$ is given by
\begin{gather*}
[\mathbf{w}_m]_{i,m+1}=
\begin{cases}
0&\text{if $1\leq i \leq m-1$}, \\
\displaystyle pT_{\beta_m\alpha_m}\prod_{k=1}^{m-l} S_{\beta_{m-k}\alpha_{m-k}}&\text{if $i =m$},\\
\displaystyle Q_{\beta_m\alpha_m}\prod_{k=1}^{m-l} S_{\beta_{m-k}\alpha_{m-k}}&\text{if $i = m+1$},\\
0&\text{if $ m+2\leq i \leq N$}.
\end{cases}
\end{gather*}
\end{Lemma}
Since we are interested in the last column of $\mathbf{A}_{\sigma}$, that is, $[\mathbf{w}_1\cdots \mathbf{w}_{N-1}]_{i,N}$, we separate $\mathbf{w}_1\cdots \mathbf{w}_{N-1} $ into two parts $(\mathbf{w}_1\cdots)(\cdots \mathbf{w}_{N-1})$ and consider the last column of the second part. The following result tells that the last column of the second part is in the form of
\begin{gather*}
\left[
 \begin{matrix}
 0 \\
 \vdots \\
 0 \\
 * \\
 \vdots \\
 *
 \end{matrix}
\right].
\end{gather*}
\begin{Lemma}\label{304pm105}
Suppose that $\mathbf{w}_{N-k},\dots, \mathbf{w}_{N-1} \neq 1$ for some $1\leq k \leq N-2$. Then,
 \begin{gather*}
 \big[\mathbf{w}_{N-k}\cdots \mathbf{w}_{N-1} \big]_{i,N} = 0
 \end{gather*}
 for all $1\leq i \leq N -k-1$, and $ \big[\mathbf{w}_{N-k}\cdots \mathbf{w}_{N-1} \big]_{i,N} \neq 0$ for $N-k \leq i\leq N$.
\end{Lemma}
\begin{proof}
If $k=1$, the statement is the same as Lemma~\ref{715pm310}. Let $2\leq k \leq N-2$. Noting that $\mathbf{w}_{N-2}$ is in the form of $\mathbf{G}_{N-1} \oplus S$ where $\mathbf{G}_{N-1} $ is an $(N-1) \times (N-1)$ matrix and $S$ is a nonzero scalar, and applying Lemma~\ref{715pm310} to $\mathbf{G}_{N-1} $, we obtain $\big[\mathbf{w}_{N-2}\big]_{i,N-1} = 0$ for all $i \neq N-1,N-2$, and $[\mathbf{w}_{N-2}]_{i,N} = 0$ for all $i \neq N$. Hence, computing $\mathbf{w}_{N-2} \mathbf{w}_{N-1}$ directly, we obtain
 \begin{gather*}
[\mathbf{w}_{N-2} \mathbf{w}_{N-1}]_{i,N} = 0
 \end{gather*}
 for $1\leq i\leq N-3$ and $[\mathbf{w}_{N-2} \mathbf{w}_{N-1}]_{i,N} \neq 0$ for $i=N-2,N-1,N$. Repeating this procedure, we obtain the required result.
\end{proof}

Now, we prove Theorem \ref{302pm106}. Recall that we use the convention about sub-matrices introduced in the beginning of this section.
\begin{proof}[Proof of Theorem~\ref{302pm106}]
Let $l$ be the largest integer such that $\mathbf{w}_l = 1$ and
\begin{gather*}
\mathbf{A}_{\sigma} = \mathbf{w}_{1}\cdots \mathbf{w}_{l-1}\mathbf{w}_{l+1}\cdots \mathbf{w}_{N-1}.
\end{gather*}
By Lemma \ref{304pm105}, $[\mathbf{w}_{l+1}\cdots \mathbf{w}_{N-1}]_{i,N} = 0 $ for $1\leq i\leq l$ and $\big[\mathbf{w}_{l+1}\cdots \mathbf{w}_{N-1}\big]_{i,N} \neq 0$
for $l+1\leq i \leq N$. Recalling the form of the matrix in (\ref{400pm320}), we see that all matrices $\mathbf{w}_{1},\dots, \mathbf{w}_{l-1}$ are in the form of $\mathbf{G}_l \oplus \mathbf{D}_{N-l}$ where $\mathbf{G}_l$ is a $l \times l$ matrix and $\mathbf{D}_{N-l}$ is an $(N-l) \times (N-l)$ diagonal matrix with some nonzero diagonal entries. Hence, $[(\mathbf{w}_{1}\cdots \mathbf{w}_{l-1})(\mathbf{w}_{l+1}\cdots \mathbf{w}_{N-1})]_{i,N}$ is zero for $1\leq i\leq l$ and nonzero for $l+1 \leq i \leq N$. Now, suppose that $w_i \neq 1$ for all $i$ so that $\mathbf{A}_{\sigma} = \mathbf{w}_{1}\cdots \mathbf{w}_{N-1}$. By Lemma~\ref{304pm105}, $[\mathbf{w}_{2}\cdots \mathbf{w}_{N-1}]_{1,N}$ is zero and $[\mathbf{w}_{2}\cdots \mathbf{w}_{N-1}]_{i,N} \neq 0$ for all $i\neq 1$, and we note that
\begin{gather*}
\mathbf{w}_1= \mathbf{T}_1(\beta,\alpha) = \left[
 \begin{matrix}
 P_{\beta\alpha} & pT_{\beta\alpha} \\
 qT_{\beta\alpha} & Q_{\beta\alpha}
 \end{matrix}
\right] \oplus \mathbf{D}_{N-2}
\end{gather*}
 for some $(\beta,\alpha)$. Hence, $\big[\mathbf{w}_{1}\cdots \mathbf{w}_{N-1}\big]_{i,N} \neq 0$ for all $i$.
\end{proof}

\subsection[Case \protect{[A\_\{sigma\}]\_\{pi,nu\} not= 0} where pi not= nu = 2...21]{Case $\boldsymbol{[\mathbf{A}_{\sigma}]_{\pi,\nu}\neq 0}$ where $\boldsymbol{\pi\neq \nu = 2\cdots 21}$}\label{329am109}

Lemma \ref{304pm105} stated that $[\mathbf{w}_{l+1}\cdots \mathbf{w}_{N-1}]_{i,N} = 0$ for $i = 1,\dots, l$ and $[\mathbf{w}_{l+1}\cdots \mathbf{w}_{N-1} ]_{i,N} \neq 0$ for $i = l+1,\dots, N$. Now, we give the formulas for these nonzero terms.
\begin{Theorem}\label{427pm106}
Let $\sigma = w_{l+1}\cdots w_{N-1}$ be expressed as in Theorem~{\rm \ref{700pm410}} with $w_{l+1}, \dots, w_{N-1} \neq 1$, and let
\begin{gather*}
\mathbf{A}_{\sigma} = \mathbf{w}_{l+1}\cdots\mathbf{w}_{N-1}
\end{gather*}
be the matrix corresponding to $\sigma$.
Then,
\begin{gather}\label{325am106}
[\mathbf{A}_{\sigma}]_{i,N} = \prod_{k=l+1}^{i-1}[\mathbf{w}_k]_{i,i}\prod_{k=i}^{N-1}[\mathbf{w}_k]_{k,k+1}
\end{gather}
for $i=l+1,\dots, N-1$.
\end{Theorem}
\begin{proof}
We prove by induction on $N$. When $N=2$, the statement obviously holds. Suppose that the statement holds for $N$, that is, for the multi-set $[1,\underbrace{2,\dots,2}_{N-1}]$. We will show that the statement is true for $N+1$, that is, for the multi-set $[1,\underbrace{2,\dots,2}_{N}]$. Let $\sigma' = w'_{l+1}\cdots w'_N$ be given as in Theorem \ref{700pm410} with $S_{N+1}$. If $T_i$ is a simple transposition in $S_N$ and $T'_i$ is a simple transposition in $S_{N+1}$, then $\mathbf{T}_i(\beta,\alpha) \oplus S_{\beta\alpha} = \mathbf{T}'_i(\beta,\alpha)$ for $i=1,\dots, N-1$. Hence, $\mathbf{w}_i\oplus {\mathcal S}_i = \mathbf{w}'_i$ for $i=1,\dots, N-1$ for some scalar ${\mathcal S}_i$, and
 so
\begin{align*}
\mathbf{A}_{\sigma'} & = \mathbf{w}'_{l+1}\cdots \mathbf{w}'_N = \mathbf{w}'_{l+1}\cdots \mathbf{w}'_{N-1}\mathbf{w}'_N \\
&=\big( (\mathbf{w}_{l+1}\cdots \mathbf{w}_{N-1} )\oplus ({\mathcal S}_{l+1}\cdots {\mathcal S}_{N-1}) \big)\mathbf{w}'_N,\qquad l+1 \leq N-1.
\end{align*}
Also, note that
 if $1\leq m,n\leq N$, then
\begin{gather*}
[ \mathbf{w}_i\oplus {\mathcal S}_i]_{m,n} = [ \mathbf{w}_i]_{m,n}= [\mathbf{w}'_i]_{m,n},\qquad i=1,\dots, N-1,
\end{gather*}
and $[\mathbf{w}'_N]_{i,N+1} = 0$ for all $1\leq i \leq N-1$ by Lemma~\ref{213am106}. Hence, for $i=1,\dots, N-1$,
\begin{align*}
 [\mathbf{A}_{\sigma'} ]_{i,N+1}& = \sum_{k=1}^{N+1} \big( (\mathbf{w}_{l+1}\cdots \mathbf{w}_{N-1} )\oplus ({\mathcal S}_{l+1}\cdots {\mathcal S}_{N-1}) \big)_{i,k} [\mathbf{w}'_N ]_{k,N+1} \\
& = \big((\mathbf{w}_{l+1}\cdots \mathbf{w}_{N-1})\oplus ({\mathcal S}_{l+1}\cdots {\mathcal S}_{N-1}) \big)_{i,N} [\mathbf{w}'_N ]_{N,N+1} \\
& = (\mathbf{w}_{l+1}\cdots \mathbf{w}_{N-1})_{i,N}[\mathbf{w}'_N]_{N,N+1}.
\end{align*}
By the induction hypothesis,
\begin{align*}
[\mathbf{w}_{l+1}\cdots \mathbf{w}_{N-1}]_{i,N}[\mathbf{w}'_N]_{N,N+1} &=\Bigg( \prod_{k=l+1}^{i-1}[\mathbf{w}_k]_{i,i}\prod_{k=i}^{N-1}[\mathbf{w}_k]_{k,k+1}\Bigg) [\mathbf{w}'_N ]_{N,N+1}\\
& = \Bigg( \prod_{k=l+1}^{i-1} [\mathbf{w}'_k ]_{i,i}\prod_{k=i}^{N-1} [\mathbf{w}'_k ]_{k,k+1}\Bigg) [\mathbf{w}'_N ]_{N,N+1}\\
& = \prod_{k=l+1}^{i-1} [\mathbf{w}'_k ]_{i,i}\prod_{k=i}^{N} [\mathbf{w}'_k ]_{k,k+1}.
\end{align*}
for $i = l+1,\dots, N-1$. Hence, we obtained
\begin{gather*}
[\mathbf{A}_{\sigma'}]_{i,N+1} = \prod_{k=l+1}^{i-1}[\mathbf{w}'_k ]_{i,i}\prod_{k=i}^{N}[\mathbf{w}'_k ]_{k,k+1}
\end{gather*}
for $i = l+1,\dots, N-1$.
If $i=N$, then using that $\mathbf{w'}_{l+1}\cdots \mathbf{w'}_{N-1} = \mathbf{w}_{l+1}\cdots \mathbf{w}_{N-1} \oplus {\mathcal S}$ for some scalar~${\mathcal S}$, Theorem \ref{1013pm104}, Lemmas~\ref{213am106} and~\ref{304pm105}, we obtain
\begin{align*}
[\mathbf{A}_{\sigma'}]_{N,N+1} & = \sum_{k=1}^{N+1} [\mathbf{w'}_{l+1}\cdots \mathbf{w'}_{N-1} ]_{N,k} [\mathbf{w}'_N ]_{k,N+1}\\
& = [\mathbf{w'}_{l+1}\cdots \mathbf{w'}_{N-1} ]_{N,N} [\mathbf{w}'_N ]_{N,N+1} = [\mathbf{w}_{l+1}\cdots \mathbf{w}_{N-1} ]_{N,N} [\mathbf{w}'_N ]_{N,N+1}\\
& = \Bigg(\prod_{k=l+1}^{N-1}[\mathbf{w}_k]_{N,N}\Bigg) [\mathbf{w}'_N ]_{N,N+1} = \Bigg(\prod_{k=l+1}^{N-1} [\mathbf{w}'_k ]_{N,N}\Bigg) [\mathbf{w}'_N ]_{N,N+1}.\tag*{\qed}
\end{align*}\renewcommand{\qed}{}
\end{proof}

Theorem \ref{1013pm104}, Theorem \ref{427pm106} and the convention on the product introduced in Section~\ref{1126pm1010} imply that~(\ref{325am106}) actually holds for $i=N$.
\begin{Corollary}\label{11301pm1010}
Let $\sigma = w_{l+1}\cdots w_{N-1}$ be expressed as in Theorem~{\rm \ref{700pm410}} with $w_{l+1}, \dots, w_{N-1} \neq 1$, and let
\[\mathbf{A}_{\sigma} = \mathbf{w}_{l+1}\cdots\mathbf{w}_{N-1}
\]
be the matrix corresponding to~$\sigma$.
Then,
\begin{gather*}
[\mathbf{A}_{\sigma}]_{i,N} = \prod_{k=l+1}^{i-1}[\mathbf{w}_k]_{i,i}\prod_{k=i}^{N-1}[\mathbf{w}_k]_{k,k+1}
\end{gather*}
for $i=l+1,\dots, N$.
\end{Corollary}
The proof of Theorem \ref{312am109} is based on Corollary~\ref{11301pm1010}.
\begin{proof}[Proof of Theorem~\ref{312am109}]
Let $l$ be the largest integer such that $w_l = 1$ and
\begin{gather*}
\mathbf{A}_{\sigma} = \big(\mathbf{w}_1\cdots \mathbf{w}_{l-1}\big)\big(\mathbf{w}_{l+1}\cdots \mathbf{w}_{N-1}\big).
\end{gather*}
Then,
\begin{gather*}
\big[\mathbf{w}_{l+1}\cdots \mathbf{w}_{N-1}\big]_{N,N} = \big[\mathbf{w}_{l+1}\big]_{N,N}\cdots \big[\mathbf{w}_{N-1}\big]_{N,N}
\end{gather*}
by Theorem \ref{1013pm104}, and
\begin{gather*}
[\mathbf{w}_{l+1}\cdots \mathbf{w}_{N-1}]_{i,N} = \prod_{k=l+1}^{i-1}[\mathbf{w}_k]_{i,i}\prod_{k=i}^{N-1}[\mathbf{w}_k]_{k,k+1}
\end{gather*}
for $i = l+1,\dots, N$ by Corollary \ref{11301pm1010}, and
\begin{gather*}
[\mathbf{w}_{l+1}\cdots \mathbf{w}_{N-1}]_{i,N} = 0
\end{gather*}
for $i = 1,\dots, l$ by Theorem~\ref{302pm106}. Note that each $\mathbf{w}_i$ for $i=1,\dots, l-1$ is written as $\mathbf{G}_{l} \oplus \mathbf{D}_{N-l}$ for some $l\times l$ matrix $\mathbf{G}_{l}$ and $(N-l) \times (N-l)$ diagonal matrix $\mathbf{D}_{N-l}$. Hence, $(\mathbf{w}_1\cdots \mathbf{w}_{l-1})$ is also in the same form as $\mathbf{G}_{l} \oplus \mathbf{D}_{N-l}$ and
\begin{gather*}
[\mathbf{w}_1\cdots \mathbf{w}_{l-1}]_{i,i} =[\mathbf{w}_1]_{i,i}\cdots [\mathbf{w}_{l-1}]_{i,i}
\end{gather*}
for $i= l+1,\dots, N$.
 Therefore, for $i = l+1,\dots, N$, noting that $[\mathbf{w}_l]_{i,i} = 1$,
\begin{align*}
[\mathbf{A}_{\sigma}]_{i,N} &= \sum_{k=1}^N [\mathbf{w}_1\cdots \mathbf{w}_{l-1} ]_{i,k} [\mathbf{w}_{l+1}\cdots \mathbf{w}_{N-1} ]_{k,N}
= [\mathbf{w}_1\cdots \mathbf{w}_{l-1} ]_{i,i} [\mathbf{w}_{l+1}\cdots \mathbf{w}_{N-1} ]_{i,N} \\
& = [\mathbf{w}_1 ]_{i,i}\cdots [\mathbf{w}_{l-1} ]_{i,i}\prod_{k=l+1}^{i-1}[\mathbf{w}_k]_{i,i}\prod_{k=i}^{N-1}[\mathbf{w}_k]_{k,k+1}
 =\prod_{k=1}^{i-1}[\mathbf{w}_k]_{i,i}\prod_{k=i}^{N-1}[\mathbf{w}_k]_{k,k+1}.
\end{align*}
If $i\leq l$, then $[\mathbf{A}_{\sigma}]_{i,N}$ must be zero by Theorem~\ref{302pm106}. If $i\leq l$, then there exists a factor $[\mathbf{w}_l]_{l,l+1}$ in~(\ref{226am109}) but $[\mathbf{w}_l]_{l,l+1}=0$ because $\mathbf{w}_l$ is the identity matrix. Hence, (\ref{226am109}) holds for all $i=1,\dots, N$. Finally, if there is no integer $l$ with $w_l=1$, we just set $l=0$ in Theorem~\ref{427pm106} to complete the proof.
\end{proof}

The explicit formulas of $[\mathbf{w}_k]_{i,i}$ and $[\mathbf{w}_k]_{k,k+1}$ in (\ref{226am109}) and in (\ref{325am106}) are provided in Proposition~\ref{459pm106}.
\begin{proof}[Proof of Proposition \ref{459pm106}]
 Recalling the form of the matrix in (\ref{400pm320}), we can easily obtain (\ref{407am106})~and~(\ref{1234am108}) by directly performing matrix multiplication in
 \begin{gather*}
 \mathbf{w}_k=\mathbf{T}_{k}(\beta_{k_1},\alpha_{k_1}) \mathbf{T}_{k-1}(\beta_{k_2},\alpha_{k_2})\cdots \mathbf{T}_{k-l+1}(\beta_{k_l},\alpha_{k_l}).\tag*{\qed}
 \end{gather*}\renewcommand{\qed}{}
\end{proof}

\appendix

\section{Alternate approach}\label{appendix-A}
\subsection{Physical interpretation}
We give a physical interpretation for the formulas given in this paper and provide an alternate approach to find $[\mathbf{A}_{\sigma}]_{\pi,\nu}$. The cardinality of the set of all permutations of the multi-set $\mathcal{M}_N =[1,2,\dots,2]$ or $\mathcal{N}_N = [1,\dots,1,2]$ is $N$. Let $F$ be the function field of all rational functions of $N$ variables $\xi_1,\dots,\xi_N \in \mathbb{C}$ over $\mathbb{C}$. Let us consider the vector space $F^N$ over the field $F$ on which a bilinear form $\langle \cdot\,|\,\mathbf{A}\,|\,\cdot \rangle\colon F^N \times F^N \to F$ is defined by
\begin{gather*}
\langle (f_1,\dots, f_N)\,|\,\mathbf{A}\,|\, (g_1,\dots, g_N) \rangle = (f_1,\dots, f_N) \mathbf{A} (g_1,\dots, g_N)^t
\end{gather*}
for any $N \times N$ matrix $\mathbf{A}$ of rational functions of $N$ variables $\xi_1,\dots,\xi_N \in \mathbb{C}$.
We identify each permutation $\nu$ of $\mathcal{M}$ as a vector in the standard basis of~$F^N$. (Here, $\mathcal{M}$ is either~$\mathcal{M}_N$ or~$\mathcal{N}_N$.) In case of $\mathcal{M} = [1,2,\dots, 2]$, we identify the permutation $12\cdots2$ as $(1,0,\dots,0)$, $212\cdots2$ as $(0,1,0,\dots,0)$ and so on. In case of $\mathcal{M} = [1,\dots,1, 2]$, we identify the permutation $1\cdots12$ as $(1,0,\dots,0)$, $1\cdots121$ as $(0,1,0,\dots,0)$ and so on. Using the \textit{bra} and the \textit{ket} notations in physics, we write $|~1~\rangle = (1,0,\dots,0)$, $|~2~\rangle = (0,1,\dots,0)$ and so on. The matrix $\mathbf{R}_{\beta\alpha}$ in (\ref{625pm72443}) is interpreted as so called \textit{the $S$-matrix $($scattering matrix$)$} in physics.

\subsubsection{Revisit to Theorem \ref{1013pm104}, Propositions \ref{323am109}, \ref{1122pm104}, and \ref{101149pm104}}
Recall (\ref{400pm320}) and (\ref{1055pm44}). We will omit the superscripts $\mathcal{M}_N$ and $\mathcal{N}_N$ in $\mathbf{A}^{\mathcal{M}_N}$ and $\mathbf{A}^{\mathcal{N}_N}$ for convenience as in Section~\ref{920am49}.
By using the \textit{bra-ket} notation, the matrix elements of $\mathbf{T}_i(\beta,\alpha)$ are given by
\begin{gather}\label{1201am1011}
\langle \nu\,|\,\mathbf{T}_i(\beta,\alpha)\, |\,\nu \rangle =
\begin{cases}
S_{\beta\alpha}& \text{if $\nu(i) = \nu(i+1)$},\\
P_{\beta\alpha}& \text{if $\nu(i) < \nu(i+1)$},\\
Q_{\beta\alpha}& \text{if $\nu(i) > \nu(i+1)$},
\end{cases}
\end{gather}
and
\begin{gather}\label{516am1010}
\langle \nu'\,|\, \mathbf{T}_i(\beta,\alpha)\, |\,\nu \rangle =
\begin{cases}
pT_{\beta\alpha}& \text{if $\nu'(i) =1, \nu'(i+1)=2$ and $\nu(i)=2, \nu(i+1)=1$},\\
qT_{\beta\alpha}& \text{if $\nu'(i) =2, \nu'(i+1)=1$ and $\nu(i)=1, \nu(i+1)=2$},\\
0& \text{otherwise}.
\end{cases}
\end{gather}
The formula
\begin{gather*}
 [\mathbf{A}_{\sigma}]_{\nu, \nu} = [\mathbf{T}_{i_j} ]_{\nu, \nu} \cdots [\mathbf{T}_{i_m} ]_{\nu, \nu} \cdots [\mathbf{T}_{i_1} ]_{\nu, \nu},
\end{gather*}
in Theorem \ref{1013pm104}, Propositions~\ref{323am109}, \ref{1122pm104}, and~\ref{101149pm104}
 is written
\begin{gather}\label{1202am1011}
\langle \nu | \mathbf{A}_{\sigma} |\nu\rangle = \langle \nu| \mathbf{T}_{i_j} |\nu\rangle \cdots \langle \nu| \mathbf{T}_{i_m} |\nu\rangle \cdots \langle \nu| \mathbf{T}_{i_1} |\nu\rangle.
\end{gather}
(\ref{1201am1011}) and (\ref{1202am1011}) motivate us to define the following \textit{operator} which does not change a permutation of species. Let $P$ be the set of all permutations of a given multi-set $ [1,2,\dots, 2]$ or $ [1,\dots,1, 2]$, and let us denote an element of $S_N \times P$ by
\begin{gather*}
(\sigma, \nu) = \left(
 \begin{matrix}
 \sigma \\
 \nu
 \end{matrix}
 \right) =
 \left(
 \begin{matrix}
 \sigma(1)\sigma(2)\cdots\sigma(N) \\
 \nu(1)\nu(2)\cdots\nu(N) \\
 \end{matrix}
 \right).
\end{gather*}
For given simple transposition $T_i$, define a mapping ${T}^*_i$ on the set of all objects written
\begin{gather*}
R \left(
 \begin{matrix}
 \sigma(1)\sigma(2)\cdots\sigma(N) \\
 \nu(1)\nu(2)\cdots\nu(N)
 \end{matrix}
 \right),
\end{gather*}
 where $R$ is 1 or a product of factors in the form of $S_{\beta\alpha}$, $P_{\beta\alpha}$, $Q_{\beta\alpha}$, $pT_{\beta\alpha}$, $qT_{\beta\alpha}$ by
 \begin{gather*}
{T}^*_i R\left(
 \begin{matrix}
 \sigma(1) & \cdots & \sigma(i) & \sigma(i+1) & \cdots & \sigma(N) \\
 \nu(1) & \cdots & \nu(i) & \nu(i+1) & \cdots & \nu(N)
 \end{matrix}
 \right)
 \\
\qquad{} = \begin{cases}
 R S_{\sigma(i+1)\sigma(i)} \left(
 \begin{matrix}
 \sigma(1) & \cdots & \sigma(i+1) & \sigma(i) & \cdots & \sigma(N) \\
 \nu(1) & \cdots & \nu(i) & \nu(i+1) & \cdots & \nu(N)
 \end{matrix}
 \right)&\text{if $\nu(i) = \nu({i+1})$},\vspace{1mm}\\
 R Q_{\sigma(i+1)\sigma(i)}\left(
 \begin{matrix}
 \sigma(1) & \cdots & \sigma(i+1) & \sigma(i) & \cdots & \sigma(N) \\
 \nu(1) & \cdots & \nu(i) & \nu(i+1) & \cdots & \nu(N) \\
 \end{matrix}
 \right)&\text{if $\nu(i) > \nu({i+1})$},\vspace{1mm}\\
 R P_{\sigma(i+1)\sigma(i)}\left(
 \begin{matrix}
 \sigma(1) & \cdots & \sigma(i+1) & \sigma(i) & \cdots & \sigma(N) \\
 \nu(1) & \cdots & \nu(i) & \nu(i+1) & \cdots & \nu(N)
 \end{matrix}
 \right)&\text{if $\nu(i) < \nu({i+1})$}.
 \end{cases}
 \end{gather*}
In other words, $T_i^*$ acts as the usual simple transposition on permutations in $S_N$ but it acts as the identity on permutations of a multi-set. Then, $[\mathbf{A}_{\sigma}]_{\nu, \nu}$ in Theorem~\ref{1013pm104}, Propositions~\ref{323am109},~\ref{1122pm104} and~\ref{101149pm104} are read off as a \textit{by-product} obtained after acting $T^*_{i_j},\dots, T^*_{i_1}$ consecutively on $(12{\cdots}N, \nu)$, that is, $\prod_{(\beta,\alpha)}R_{\beta\alpha}$ in
 \begin{gather*}
 T^*_{i_j}\cdots T^*_{i_1}\left(
 \begin{matrix}
 1 & \cdots & N-1 & N \\
 \nu(1) & \nu(2) & \cdots & \nu(N)
 \end{matrix}
 \right) = \bigg(\prod_{(\beta,\alpha)}R_{\beta\alpha}\bigg)\left(
 \begin{matrix}
 \sigma(1) & \cdots & \sigma(N-1) & \sigma(N) \\
 \nu(1) & \nu(2) & \cdots & \nu(N)\\
 \end{matrix}
 \right).
 \end{gather*}

\subsection{Revisit to Theorem \ref{312am109}}
The formula
\begin{gather*}
[\mathbf{A}_{\sigma}]_{i,N} =\prod_{k=1}^{i-1}[\mathbf{w}_k]_{i,i}\prod_{k=i}^{N-1}[\mathbf{w}_k]_{k,k+1}
\end{gather*}
in Theorem \ref{312am109} is written by using the \textit{bra-ket} notation
\begin{gather*}
\langle i| \mathbf{A}_{\sigma} |N\rangle = \langle i| \mathbf{w}_{1} |i\rangle \cdots \langle i| \mathbf{w}_{i-1} |i\rangle \langle i| \mathbf{w}_{i} |i+1\rangle \cdots \langle N-1| \mathbf{w}_{N-1} |N\rangle.
\end{gather*}
We observe that each $\mathbf{w}_{k}$ with $k=i, \dots, N-1$ changes $|k+1\rangle$ to $|k\rangle$ but each $\mathbf{w}_{k}$ with $k=1,\dots, i-1$ does not change $|i\rangle$. Also, we observe that
\begin{gather*}
[\mathbf{w}_k]_{k,k+1}= \langle k | \mathbf{w}_{k} |k+1\rangle = \langle k| \mathbf{T}_{k} |k+1\rangle \langle k+1| \mathbf{T}_{k-1} |k+1\rangle \cdots \langle k+1| \mathbf{T}_{k-l+1} |k+1\rangle
\end{gather*}
 by (\ref{1234am108}), (\ref{1201am1011}) and (\ref{516am1010}).
 Motivated by these observations, let us define the following \textit{operator} which changes a permutation of species represented by $|k+1\rangle$ to a permutation of species represented by $|k\rangle$ and vice versa to consider all four cases of $\nu=2\cdots 21, 12\cdots 2, 21\cdots 1, 1\cdots 12.$ For given simple transposition~$T_k$, define a mapping $\hat{T}_k$ on the set of all objects written
\begin{gather*}
R \left(
 \begin{matrix}
 \sigma(1)\sigma(2)\cdots\sigma(N) \\
 \nu(1)\nu(2)\cdots\nu(N)
 \end{matrix}
 \right),
\end{gather*}
 where $R$ is 1 or a product of factors in the form of $S_{\beta\alpha}$, $P_{\beta\alpha}$, $Q_{\beta\alpha}$, $pT_{\beta\alpha}$, $qT_{\beta\alpha}$ by
 \begin{gather*}
\hat{T}_i R\left(
 \begin{matrix}
 \sigma(1) & \cdots & \sigma(i) & \sigma(i+1) & \cdots & \sigma(N) \\
 \nu(1) & \cdots & \nu(i) & \nu(i+1) & \cdots & \nu(N)
 \end{matrix}
 \right)
 \\
\qquad{} = \begin{cases}
 R pT_{\sigma(i+1)\sigma(i)}\left(
 \begin{matrix}
 \sigma(1) & \cdots & \sigma(i+1) & \sigma(i) & \cdots & \sigma(N) \\
 \nu(1) & \cdots & \nu(i+1) & \nu(i) & \cdots & \nu(N)
 \end{matrix}
 \right)&\text{if $\nu(i) > \nu({i+1})$},\vspace{1mm}\\
 R qT_{\sigma(i+1)\sigma(i)}\left(
 \begin{matrix}
 \sigma(1) & \cdots & \sigma(i+1) & \sigma(i) & \cdots & \sigma(N) \\
 \nu(1) & \cdots & \nu(i+1) & \nu(i) & \cdots & \nu(N)
 \end{matrix}
 \right)&\text{if $\nu(i) < \nu({i+1})$},\vspace{1mm}\\
 R S_{\sigma(i+1)\sigma(i)} \left(
 \begin{matrix}
 \sigma(1) & \cdots & \sigma(i+1) & \sigma(i) & \cdots & \sigma(N) \\
 \nu(1) & \cdots & \nu(i) & \nu(i+1) & \cdots & \nu(N)
 \end{matrix}
 \right)&\text{if $\nu(i) = \nu({i+1})$}.
 \end{cases}
 \end{gather*}
With the argument in the above, we reformulate Theorem \ref{312am109} in terms of the operators~$T^*$ and~$\hat{T}$ as follows:
\begin{Corollary}\label{235am1011}
Let
\begin{gather*}
\sigma = w_1\cdots w_{N-1}= \underbrace{T_{k_j}\cdots T_{k_{m-1}}}_{w_1\cdots w_{i-1}} \underbrace{T_{k_{m}}\cdots T_{k_1}}_{w_i\cdots w_{N-1}}
\end{gather*}
be expressed as in Theorem~{\rm \ref{700pm410}}. If $w_i,\dots, w_{N-1} \neq 1$, then
\begin{gather*}
T^*_{k_j}\cdots T^*_{k_{m-1}}\hat{T}_{k_{m}}\cdots \hat{T}_{k_1}\left(
 \begin{matrix}
 1~2~\cdots ~N \\
 \nu(1)\cdots \nu(N)
 \end{matrix}
 \right) = [\mathbf{A}_{\sigma}]_{\pi^{(i)},\nu} \left(
 \begin{matrix}
 \sigma(1) \cdots \sigma(N) \\
 \pi^{(i)}(1)\cdots \pi^{(i)}(N)
 \end{matrix}
 \right),
\end{gather*}
and if $i\leq l$ with $w_l =1$, then
$ [\mathbf{A}_{\sigma}]_{\pi^{(i)},\nu} = 0$.
\end{Corollary}
\begin{Example}\label{317pm127}
Suppose that we want to find $\big[\mathbf{A}_{4321}\big]_{2212, 2221}$. By Theorem \ref{700pm410}, we have
\begin{gather*}
4321 = \underbrace{(T_1)}_{w_1} \underbrace{(T_2T_1)}_{w_2} \underbrace{(T_3T_2T_1)}_{w_3}.
\end{gather*}
Figure \ref{fig:M444} shows how Corollary \ref{235am1011} is used to find
\begin{gather*}
 \big[A_{4321} \big]_{2212,2221}=S_{43}Q_{42}S_{32}pT_{41}S_{31}S_{21}.
\end{gather*}

\begin{figure}[h!]\centering
\begin{tikzpicture}
\draw [thick] (0.5,0) node[]{$1$} ;
\draw [thick] (1,0) node[]{$2$} ;
\draw [thick] (1.6,0) node[]{$3$} ;
\draw [thick] (2.1,0) node[]{$4$} ;
\draw[thick] (0.5, - 1) node[]{\circled{2}} ;
\draw[thick] (1, - 1) node[]{\circled{2}} ;
\draw[thick] (1.6, - 1) node[]{\circled{2}} ;
\draw[thick] (2.1, - 1) node[]{\circled{1}} ;
\draw [dashed, cyan] (0.1, 0.5) -- (1.3,0.5);
\draw [dashed, cyan] (0.1, -1.5) -- (1.3,-1.5);
\draw [dashed, cyan] (0.1, 0.5) -- (0.1,-1.5);
\draw [dashed, cyan] (1.3, 0.5) -- (1.3,-1.5);

\draw [thick, red] (0, 0.3) -- (2.4,0.3);
\draw [thick,red] (0, -1.3) -- (2.4,-1.3);
\draw [thick, red] (0, 0.3) -- (0,-1.3);
\draw [thick, red] (2.4, 0.3) -- (2.4,-1.3);

\draw[thick, ->] (2.6, -0.5) -- (3.6, -0.5) ;
\node at (3.1,-0.2) {$\hat{T}_1$};
\node at (3.1,-0.8) {$S_{21}$};

\draw [thick] (4.1,0) node[]{$2$} ;
\draw [thick] (4.7,0) node[]{$1$} ;
\draw [thick] (5.2,0) node[]{$3$} ;
\draw [thick] (5.8,0) node[]{$4$} ;
\draw[thick] (4.1, - 1) node[]{\circled{2}} ;
\draw[thick] (4.7, - 1) node[]{\circled{2}} ;
\draw[thick] (5.2, - 1) node[]{\circled{2}} ;
\draw[thick] (5.8, - 1) node[]{\circled{1}} ;
\draw [dashed, cyan] (4.4, 0.5) -- (5.5,0.5);
\draw [dashed, cyan] (4.4, -1.5) -- (5.5,-1.5);
\draw [dashed, cyan] (4.4, 0.5) -- (4.4, -1.5);
\draw [dashed, cyan] (5.5, 0.5) -- (5.5,-1.5);

\draw [thick, red] (3.8, 0.3) -- (6.1,0.3);
\draw [thick,red] (3.8, -1.3) -- (6.1,-1.3);
\draw [thick, red] (3.8, 0.3) -- (3.8,-1.3);
\draw [thick, red] (6.1, 0.3) -- (6.1,-1.3);

\draw[thick, ->] (6.3, -0.5) -- (7.3, -0.5) ;
\node at (6.8,-0.2) {$\hat{T}_2$};
\node at (6.8,-0.8) {$S_{31}$};

\draw [thick] (7.9,0) node[]{$2$} ;
\draw [thick] (8.4,0) node[]{$3$} ;
\draw [thick] (9,0) node[]{$1$} ;
\draw [thick] (9.5,0) node[]{$4$} ;
\draw[thick] (7.9, - 1) node[]{\circled{2}} ;
\draw[thick] (8.4, - 1) node[]{\circled{2}} ;
\draw[thick] (9, - 1) node[]{\circled{2}} ;
\draw[thick] (9.5, - 1) node[]{\circled{1}} ;
\draw [dashed, cyan] (8.7, 0.6) -- (9.8,0.6);
\draw [dashed, cyan] (8.7, -1.6) -- (9.8,-1.6);
\draw [dashed, cyan] (8.7, 0.6) -- (8.7,-1.6);
\draw [dashed, cyan] (9.8, 0.6) -- (9.8,-1.6);

\draw [thick, red] (7.6, 0.3) -- (9.9,0.3);
\draw [thick,red] (7.6, -1.3) -- (9.9,-1.3);
\draw [thick, red] (7.6, 0.3) -- (7.6,-1.3);
\draw [thick, red] (9.9, 0.3) -- (9.9,-1.3);

\draw[thick, ->] (10.1, -0.5) -- (11.1, -0.5) ;
\node at (10.6,-0.2) {$\hat{T}_3$};
\node at (10.6,-0.8) {$pT_{41}$};

\draw [thick] (11.8,0) node[]{$2$} ;
\draw [thick] (12.3,0) node[]{$3$} ;
\draw [thick] (13,0) node[]{$4$} ;
\draw [thick] (13.5,0) node[]{$1$} ;
\draw[thick] (11.8, - 1) node[]{\circled{2}} ;
\draw[thick] (12.3, - 1) node[]{\circled{2}} ;
\draw[thick] (13, - 1) node[]{\circled{1}} ;
\draw[thick] (13.5, - 1) node[]{\circled{2}} ;
\draw [dashed, cyan] (11.5, 0.6) -- (12.6,0.6);
\draw [dashed, cyan] (11.5, -1.6) -- (12.6,-1.6);
\draw [dashed, cyan] (11.5, 0.6) -- (11.5,-1.6);
\draw [dashed, cyan] (12.6, 0.6) -- (12.6,-1.6);

\draw [thick, red] (11.4, 0.3) -- (13.8,0.3);
\draw [thick,red] (11.4, -1.3) -- (13.8,-1.3);
\draw [thick, red] (11.4, 0.3) -- (11.4,-1.3);
\draw [thick, red] (13.8, 0.3) -- (13.8,-1.3);

\draw[thick, ->] (14.1, -0.5) -- (14.3, -0.5) -- (14.3, -3) -- (13.3, -3) ;
\node at (13.8,-2.7) {$T^*_1$};
\node at (13.8,-3.3) {$S_{32}$};

\draw [thick] (12.8,-2.5) node[]{$1$} ;
\draw [thick] (12.2,-2.5) node[]{$4$} ;
\draw [thick] (11.7,-2.5) node[]{$2$} ;
\draw [thick] (11.1,-2.5) node[]{$3$} ;
\draw[thick] (12.8, - 3.5) node[]{\circled{2}} ;
\draw[thick] (12.2, - 3.5) node[]{\circled{1}} ;
\draw[thick] (11.7, - 3.5) node[]{\circled{2}} ;
\draw[thick] (11.1, - 3.5) node[]{\circled{2}} ;
\draw [dashed, cyan] (11.4, -2.1) -- (12.5,-2.1);
\draw [dashed, cyan] (11.4, -4.1) -- (12.5,-4.1);
\draw [dashed, cyan] (11.4, -2.1) -- (11.4,-4.1);
\draw [dashed, cyan] (12.5, -2.1) -- (12.5,-4.1);

\draw [thick, red] (10.7, -2.2) -- (13.1,-2.2);
\draw [thick,red] (10.7, -3.8) -- (13.1,-3.8);
\draw [thick, red] (10.7, -2.2) -- (10.7,-3.8);
\draw [thick, red] (13.1, -2.2) -- (13.1,-3.8);

\draw[thick, ->] (10.4, -3) -- (9.4, -3) ;
\node at (9.9,-2.7) {$T^*_2$};
\node at (9.9,-3.3) {$Q_{42}$};

\draw [thick] (8.9,-2.5) node[]{$1$} ;
\draw [thick] (8.4,-2.5) node[]{$2$} ;
\draw [thick] (7.8,-2.5) node[]{$4$} ;
\draw [thick] (7.3,-2.5) node[]{$3$} ;
\draw[thick] (8.9, - 3.5) node[]{\circled{2}} ;
\draw[thick] (8.4, - 3.5) node[]{\circled{1}} ;
\draw[thick] (7.8, - 3.5) node[]{\circled{2}} ;
\draw[thick] (7.3, - 3.5) node[]{\circled{2}} ;
\draw [dashed, cyan] (7, -2.1) -- (8.1,-2.1);
\draw [dashed, cyan] (7, -4.1) -- (8.1,-4.1);
\draw [dashed, cyan] (7, -2.1) -- (7,-4.1);
\draw [dashed, cyan] (8.1, -2.1) -- (8.1,-4.1);

\draw [thick, red] (6.9, -2.2) -- (9.2,-2.2);
\draw [thick,red] (6.9, -3.8) -- (9.2,-3.8);
\draw [thick, red] (6.9, -2.2) -- (6.9,-3.8);
\draw [thick, red] (9.2, -2.2) -- (9.2,-3.8);

\draw[thick, ->] (6.6, -3) -- (5.6, -3) ;
\node at (6.1,-2.7) {$T^*_1$};
\node at (6.1,-3.3) {$S_{43}$};

\draw [thick] (3.6,-2.5) node[]{$4$} ;
\draw [thick] (4.1,-2.5) node[]{$3$} ;
\draw [thick] (4.6,-2.5) node[]{$2$} ;
\draw [thick] (5.1,-2.5) node[]{$1$} ;
\draw[thick] (3.6, - 3.5) node[]{\circled{2}} ;
\draw[thick] (4.1, - 3.5) node[]{\circled{2}} ;
\draw[thick] (4.6, - 3.5) node[]{\circled{1}} ;
\draw[thick] (5.1, - 3.5) node[]{\circled{2}} ;

\draw [thick, red] (5.4, -2.2) -- (3.3,-2.2);
\draw [thick,red] (5.4, -3.8) -- (3.3,-3.8);
\draw [thick, red] (5.4, -2.2) -- (5.4,-3.8);
\draw [thick, red] (3.3, -2.2) -- (3.3,-3.8);

\end{tikzpicture}
\caption{$[\textbf{A}_{4321}]_{2212,2221}=S_{43}Q_{42}S_{32}pT_{41}S_{31}S_{21}$.}\label{fig:M444}
\end{figure}
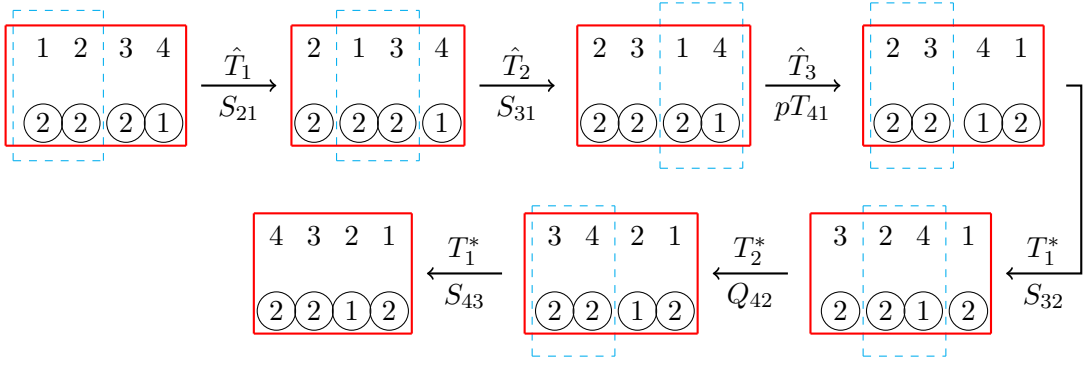
\end{Example}

\section[Matrix elements \protect{[A\_\{sigma\}]\_\{pi,2221\}}]{Matrix elements $\boldsymbol{[\mathbf{A}_{\sigma}]_{\pi,2221}}$}\label{715pm44}

\begin{table}[h!]\renewcommand{\arraystretch}{1.49}\centering\small
\begin{tabular}{|@{\,}c@{\,}|@{\,}c@{\,}|@{\,}c@{\,}|@{\,}c@{\,}|@{\,}c@{\,}|}
 \hline
 $\sigma$& $[\mathbf{A}_{\sigma}]_{2221,2221}$ & $[\mathbf{A}_{\sigma}]_{2212,2221}$& $[\mathbf{A}_{\sigma}]_{2122,2221}$& $[\mathbf{A}_{\sigma}]_{1222,2221}$ \\
 \hline
 1234 & $1$ & $0$ &$0$&$0$\\ \hline
 1243 &$Q_{43}$ & $pT_{43}$ & $0$&$0$\\\hline
 1324 & $S_{32}$ & $0$ & $0$&$0$\\ \hline
 1342 & $S_{32}Q_{42}$ &$S_{32}pT_{42}$& $0$&$0$\\ \hline
 1423 & $Q_{43}S_{42}$ &$pT_{43}Q_{42}$ & $pT_{43}pT_{42}$&$0$\\ \hline
 1432 & $S_{32}Q_{42}S_{43}$ & $S_{32}pT_{42}Q_{43}$ & $pT_{43}pT_{42}S_{32}$&$0$\\ \hline
 \hline
 2134 & $S_{21}$ & $0$ &$0$&$0$\\ \hline
 2143 &$S_{21}Q_{43}$ & $S_{21}pT_{43}$ & $0$&$0$\\\hline
 2314 & $S_{21}S_{31}$ & $0$ & $0$&$0$\\ \hline
 2341 & $S_{21}S_{31}Q_{41}$ &$S_{21}S_{31}pT_{41}$& $0$&$0$\\ \hline
 2413 & $S_{21}Q_{43}S_{41}$ &$S_{21}pT_{43}Q_{41}$ & $S_{21}pT_{43}pT_{41}$&$0$\\ \hline
 2431 & $S_{21}S_{31}Q_{41}S_{43}$ & $S_{21}S_{31}pT_{41}Q_{43}$ & $S_{21}pT_{43}pT_{41}S_{31}$&$0$\\ \hline
\hline
 3124 & $S_{32}S_{31}$ & $0$ &$0$&$0$\\ \hline
 3142 & $S_{32}S_{31}Q_{42}$ & $S_{32}S_{31}pT_{42}$ & $0$&$0$\\\hline
 3214 & $S_{32}S_{31}S_{21}$ & $0$ & $0$&$0$\\ \hline
 3241 & $S_{32}S_{31}S_{21}Q_{41}$ &$S_{32}S_{31}S_{21}pT_{41}$& $0$&$0$\\ \hline
 3412 & $S_{32}S_{31}Q_{42}S_{41}$ &$S_{32}S_{31}pT_{42}Q_{41}$ & $S_{32}S_{31}pT_{42}pT_{41}$&$0$\\ \hline
 3421 & $S_{32}S_{31}S_{21}Q_{41}S_{42}$ & $S_{32}S_{31}S_{21}pT_{41}Q_{42}$ & $S_{32}S_{31}pT_{42}pT_{41}S_{21}$&$0$\\ \hline
 \hline
 4123 & $Q_{43}S_{42}S_{41}$ & $pT_{43}Q_{42}S_{41}$ &$pT_{43}pT_{42}Q_{41}$&$pT_{43}pT_{42}pT_{41}$\\ \hline
 4132 & $S_{32}Q_{42}S_{43}S_{41}$ & $S_{32}pT_{42}Q_{43}S_{41}$ & $S_{32}pT_{42}pT_{43}Q_{41}$&$S_{32}pT_{42}pT_{43}pT_{41}$\\\hline
 4213 & $S_{21}Q_{43}S_{41}S_{42}$ & $S_{21}pT_{43}Q_{41}S_{42}$ & $S_{21}pT_{43}pT_{41}Q_{42}$&$S_{21}pT_{43}pT_{41}pT_{42}$\\ \hline
 4231 & $S_{21}S_{31}Q_{41}S_{43}S_{42}$ &$S_{21}S_{31}pT_{41}Q_{43}S_{42}$& $S_{21}S_{31}pT_{41}pT_{43}Q_{42}$&$S_{21}S_{31}pT_{41}pT_{43}pT_{42}$\\ \hline
 4312 & $S_{32}S_{31}Q_{42}S_{41}S_{43}$ &$S_{32}S_{31}pT_{42}Q_{41}S_{43}$ & $S_{32}S_{31}pT_{42}pT_{41}Q_{43}$&$S_{32}S_{31}pT_{42}pT_{41}pT_{43}$\\ \hline
 4321 & $S_{32}S_{31}S_{21}Q_{41}S_{42}S_{43}$ & $S_{32}S_{31}S_{21}pT_{41}Q_{42}S_{43}$ & $S_{32}S_{31}S_{21}pT_{41}pT_{42}Q_{43}$&$S_{32}S_{31}S_{21}pT_{41}pT_{42}pT_{43}$\\ \hline
\end{tabular}
\end{table}

\section[Matrix elements \protect{[A\_\{sigma\}]\_\{pi,1112\}}]{Matrix elements $\boldsymbol{[\mathbf{A}_{\sigma}]_{\pi,1112}}$}\label{7155pm44}

\begin{table}[h!]\renewcommand{\arraystretch}{1.49}\centering\small
\begin{tabular}{|@{\,}c@{\,}|@{\,}c@{\,}|@{\,}c@{\,}|@{\,}c@{\,}|@{\,}c@{\,}|}
 \hline
 $\sigma$& $[\mathbf{A}_{\sigma}]_{1112,1112}$ & $[\mathbf{A}_{\sigma}]_{1121,1112}$& $[\mathbf{A}_{\sigma}]_{1211,1112}$& $[\mathbf{A}_{\sigma}]_{2111,1112}$ \\
 \hline
 1234 & $1$ & $0$ &$0$&$0$\\ \hline
 1243 & $P_{43}$& $qT_{43}$ & $0$ & $0$\\\hline
 1324 & $S_{32}$& $0$ & $0$ & $0$ \\ \hline
 1342 & $P_{42}S_{32}$& $qT_{42}S_{32}$ & $0$ & $0$ \\ \hline
 1423 & $P_{43}S_{42}$& $qT_{43}P_{42}$ & $qT_{43}qT_{42}$ & $0$ \\ \hline
 1432 & $S_{43}P_{42}S_{32}$& $S_{43}qT_{42}P_{32}$ & $S_{43}qT_{42}qT_{32}$ & $0$ \\ \hline
 \hline
 2134 & $S_{21}$ & $0$ & $0$ & $0$ \\ \hline
 2143 & $P_{43}S_{21}$ & $qT_{43}S_{21}$ & $0$ & $0$ \\\hline
 2314 & $S_{31}S_{21}$ & $0$ & $0$ & $0$\\ \hline
 2341 & $P_{41}S_{31}S_{21}$ & $qT_{41}S_{31}S_{21}$& $0$ & $0$\\ \hline
 2413 & $P_{43}S_{41}S_{21}$ & $qT_{43}P_{41}S_{21}$ & $qT_{43}qT_{41}S_{21}$& $0$\\ \hline
 2431 & $S_{43}P_{41}S_{31}S_{21}$ & $P_{43}qT_{41}S_{31}S_{21}$ & $qT_{43}qT_{41}S_{31}S_{21}$& $0$\\ \hline
 \hline
 3124 & $S_{32}S_{31}$ & $0$ & $0$ &$0$ \\ \hline
 3142 & $P_{42}S_{32}S_{31}$ & $qT_{42}S_{32}S_{31}$ & $0$ & $0$\\\hline
 3214 & $S_{32}S_{31}S_{21}$ & $0$ & $0$ & $0$\\ \hline
 3241 & $P_{41}S_{32}S_{31}S_{21}$ & $qT_{41}S_{32}S_{31}S_{21}$& $0$ & $0$\\ \hline
 3412 & $P_{42}S_{41}S_{32}S_{31}$ & $qT_{42}P_{41}S_{32}S_{31}$ & $qT_{42}qT_{41}S_{32}S_{31}$ & $0$\\ \hline
 3421 & $S_{42}P_{41}S_{32}S_{31}S_{21}$ & $P_{42}qT_{41}S_{32}S_{31}S_{21}$ & $qT_{42}qT_{41}S_{32}S_{31}S_{21}$ & $0$\\ \hline
 \hline
 4123 & $P_{43}S_{42}S_{41}$ & $qT_{43}P_{42}S_{41}$ & $qT_{43}qT_{42}P_{41}$ &$qT_{43}qT_{42}qT_{41}$\\ \hline
 4132 & $S_{43}P_{42}S_{41}S_{32}$ & $P_{43}qT_{42}S_{41}S_{32}$ & $qT_{43}qT_{42}P_{41}S_{32}$ & $qT_{43}qT_{42}qT_{41}S_{32}$\\\hline
 4213 & $P_{43}S_{42}S_{41}S_{21}$ & $qT_{43}S_{42}P_{41}S_{21}$ & $qT_{43}P_{42}qT_{41}S_{21}$ &$qT_{43}qT_{42}qT_{41}S_{21}$\\ \hline
 4231 & $S_{43}S_{42}P_{41}S_{31}S_{21}$ & $P_{43}S_{42}qT_{41}S_{31}S_{21}$ & $qT_{43}P_{42}qT_{41}S_{31}S_{21}$ &$qT_{43}qT_{42}qT_{41}S_{31}S_{21}$\\ \hline
 4312 & $S_{43}P_{42}S_{41}S_{32}S_{31}$ & $S_{43}qT_{42}P_{41}S_{32}S_{31}$ & $P_{43}qT_{42}qT_{41}S_{32}S_{31}$ &$qT_{43}qT_{42}qT_{41}S_{32}S_{31}$\\ \hline
 4321 & $S_{43}S_{42}P_{41}S_{32}S_{31}S_{21} $& $S_{43}P_{42}qT_{41}S_{32}S_{31}S_{21}$ & $P_{43}qT_{42}qT_{41}S_{32}S_{31}S_{21}$ &$qT_{43}qT_{42}qT_{41}S_{32}S_{31}S_{21}$\\ \hline
\end{tabular}\vspace{-2mm}
\end{table}

\subsection*{Acknowledgements} This work was supported by the faculty development competitive research grants (090118FD5341 and 021220FD4251) by Nazarbayev University. We are thankful to Kamila Izhanova for assisting in preparation for the manuscript and to Francesco Sica for valuable comments. Most of all, we deeply appreciate anonymous referees for providing valuable comments to improve the earlier version of this paper.

\vspace{-2mm}

\pdfbookmark[1]{References}{ref}
\LastPageEnding

\end{document}